\documentclass[aip,jmp,preprint,a4paper]{revtex4-1}
\usepackage{cancel}
\usepackage{amsmath}
\usepackage{amssymb}
\usepackage{amsthm}
\usepackage{graphicx}
\usepackage{dcolumn}
\usepackage{bm}
\usepackage{upgreek}
\usepackage{esint}
\usepackage{amsthm}

\usepackage{commands}

\theoremstyle{definition}
\newtheorem{definition}{Definition}

\theoremstyle{remark}
\newtheorem{remark}{Remark}

\newtheorem{theorem}{Theorem}

\draft 

\begin{document}


\title{Improved accuracy in degenerate variational integrators for guiding center and magnetic field line flow} 
\author{J. W. Burby}
\affiliation{Los Alamos National Laboratory, Los Alamos, NM 87544}
\author{John M.~Finn}
\affiliation{Tibbar Plasma Technologies, Los Alamos, NM 87544}
\author{C.~Leland Ellison}
\affiliation{Lawrence Livermore National Laboratory, Livermore, CA 94550}

\begin{abstract}
First-order accurate degenerate variational integration (DVI) was introduced
in Ref.\,\onlinecite{Ellison_2018} for systems with
a degenerate Lagrangian, i.e.\,one in which the velocity-space Hessian is singular. In this paper we introducing second order accurate DVI schemes, both with and without non-uniform time stepping. We show that it is not in general possible to construct a second order scheme with a preserved two-form by composing a first order scheme with its adjoint, and discuss the conditions under which such a composition is possible. We build two classes of 
second order accurate DVI schemes. We test these second order schemes numerically on two systems having noncanonical variables, namely the magnetic field line and
guiding center systems. 
Variational integration for Hamiltonian
systems with nonuniform time steps, in terms of an extended phase
space Hamiltonian, is generalized to noncanonical variables. It is shown
that preservation of proper degeneracy leads to single-step
methods without parasitic modes, i.e.~to non-uniform
time step DVIs. This extension applies to second
order accurate as well as first order schemes, and can be applied
to adapt the time stepping to an error estimate.

\end{abstract}






\date{\today}


%

\pacs{}

\maketitle 

\section{Introduction}

A variety of dissipation-free dynamical models in plasma physics share the property of arising from variational principles, including the guiding center equations, magnetic field line flow, and collision-free Vlasov dynamics. Variational integrators discretize the action associated with this variational principle rather than discretizing the equations of motion directly. The advantages of variational integrators are similar to those of the more specialized \emph{symplectic integrators}, most of which discretize Hamiltonian systems in canonical variables. Variational integrators, especially those we study in this paper, can conveniently deal with noncanonical variables. These methods preserve exactly the Hamiltonian (or variational) nature of the original ODE system, and main advantages of these integration methods accrue when very long timescale behavior is to be studied. 
For example, when using full-orbit simulations to assess the validity of the 
guiding center approximation for runaway electrons, Liu, Wang, and Qin\cite{JLiu_2016} concluded that simulations involving approximately $10^{11}$ timesteps were required.


The task of finding a reliable variational integrator for a given variational dynamical system is generally quite challenging. The easiest case occurs when the Lagrangian underlying the variational principle is \emph{non-degenerate},\cite{Marsden_2001} meaning that the velocity-space Hessian is invertible. Stable low-order and high-order variational integrators may be constructed in the non-degenerate setting using systematic procedures. The problem becomes much more challenging, however, when the velocity space Hessian is degenerate, i.e.\,has a nontrivial null space. This sort of degeneracy arises, for instance, when dealing with a so-called phase space Lagrangian\cite{Cary_1983}, either in canonical or noncanonical variables; such a Lagrangians is \emph{linear} in the velocities, and therefore has Hessian equal to zero.

Ellison\,\emph{et.\,al.}\cite{Ellison_2018} showed that when the variational integrator formalism discussed in Ref.\,\onlinecite{Marsden_2001} is applied to a phase space Lagrangian, the resulting scheme typically performs very poorly; unphysical, potentially unstable parasitic modes arise and spoil the benefits of using a variational discretization. This flaw in the basic theory of variational integration represents a serious shortcoming as far as its applicability is concerned, because degenerate Lagrangians are commonly encountered in practice. In plasma physics phase space Lagrangians are routinely used to model the fundamental problems of magnetic field line flow\cite{Cary_1983} and guiding center motion,\cite{Littlejohn_1983} and have been shown to describe many infinite-dimensional plasma models as well.\cite{Burby_thesis_2015,Burby_two_fluid_2017,Burby_fdk_2017}

In a more optimistic vein, Refs.\,\onlinecite{Ellison_2018,Ellison_thesis} also proposed a conceptually appealing strategy for avoiding the generic pitfalls of variational integration applied to phase space Lagrangians. The idea was to select carefully the discretization of the Lagrangian so that it preserves the degeneracy of the continuous Lagrangian. Such \emph{properly degenerate} discrete Lagrangians were shown to be free of the parasitic modes that plague generic discrete phase space Lagrangians. What therefore emerged from this work was a refined notion of variational integration appropriate to phase space Lagrangians (and perhaps more general degenerate Lagrangians as well) termed \emph{degenerate variational integration} (DVI). In the same work, degenerate variational integration was applied to magnetic field line flow and to guiding center motion, the latter under the restriction that one covariant component of the magnetic field is zero. Indeed, good long-term behavior of the orbits was observed. By exploiting a near-identity transformation of standard guiding center theory, Burby and Ellison\cite{Burby_Ellison_2017}
showed that DVI can still be applied to guiding center dynamics if this stringent constraint on the magnetic field is lifted.

While degenerate variational integration may be a promising candidate for coping with degenerate Lagrangians within a variational integration framework, DVI theory as it stands today is still in its infancy. In particular, the examples of degenerate variational integrators in Refs.\,\onlinecite{Ellison_2018,Ellison_thesis} suffer from three important drawbacks. First, (a) they start from continuous dynamics formulated in terms of either canonical variables or a restricted class of noncanonical variables.
The magnetic field line and guiding center examples of Refs.\,\onlinecite{Ellison_thesis} and \onlinecite{Ellison_2018} belong to this class.
Next, (b) they only achieve first-order accuracy in time. Finally, (c) they rely on uniform time stepping. Moreover, it is presently unclear if these drawbacks in examples are reflections of inherent limitations of the DVI concept, or only apparent limitations that might be overcome with additional insights.

The first purpose of this paper is to construct a large class of \emph{second-order accurate} degenerate variational integrators,  involving the so-called processing technique.\cite{Blanes_2004}
In particular we aim to address issue (b) by formulating two related second-order accurate DVI schemes and applying them to the field line and guiding center problems considered in Refs.\,\onlinecite{Ellison_2018} and \onlinecite{Burby_Ellison_2017}. Our second purpose is to address issue (c) by formulating non-uniform time stepping for DVI. The method chosen is related to the well-known method of extended phase space,\cite{Hairer_2006} generalized to the above class of noncanonical variables. Issue (a) will be the subject of future publications.


After providing an updated discussion of the basic elements of DVI theory in Section \ref{sec:DVI}, a class of second-order accurate DVIs will be presented in Section \ref{sec:DVI2}. We begin by discussing these DVIs for systems with canonical variables, showing transparently why it is not possible in general to obtain second order accuracy by composing a first order scheme with its adjoint. We also show how to formulate these second-order systems in a restricted class of noncanonical systems. In Section \ref{MFL-GC} we formulate these schemes for the magnetic field line integration problem and the guiding center problem, both in this restricted class. In Sec.\,\ref{sec:examples} we report on the numerical application of these second-order DVI schemes to the field line and guiding center examples, showing good long-time behavior and second order accuracy. 

In Sec.~\ref{sec:Non-unifTS} we first discuss nonuniform time stepping
in the context of canonical systems, showing the extended phase space
action and discretizations of it. We discuss the condition for a single-step
scheme, i.e.~a degenerate variational integrator (DVI), as described
in  Sec.~\ref{sec:DVI}. This condition, as for uniform time
steps, is that the discrete Hessian has the proper rank, the rank
of the continuous Hessian. We proceed to show how to apply the extended
phase space method to the class of noncanonical variables described in Sec.\,\ref{sec:DVI2}. We apply this methodology
to the field line and guiding
center examples. These systems are in the special form of noncanonical
variables described in Sec.~\ref{sec:DVI}, and a method is
described to obtain a single-step (DVI) scheme system with this special
form of noncanonical variables. As before, the rank of the discrete
Hessian predicts the single-step nature of the schemes. We also argue
that it is straightforward to apply this extended phase space method
to second and higher order schemes and to adaptive time stepping.

We summarize and discuss the results of this paper in Sec.\,\ref{sec:discussion}

\section{Degenerate Variational Integration: review and recent developments\label{sec:DVI}}
The purpose of this Section is to provide details of the basic properties of DVI. While a similar discussion appears in Ref.\,\onlinecite{Ellison_2018}, the ensuing discussion will reflect an improved understanding of DVI that has developed since the publication of Ref.\,\onlinecite{Ellison_2018}.

DVI is a refinement of variational integration that applies to phase space Lagrangians, and perhaps to more general degenerate Lagrangians. It is therefore helpful to recall briefly the basic ingredients of variational integration along the lines of Ref.\,\onlinecite{Marsden_2001}. To that end, consider a dynamical system governed by the variational principle based on the action $S$,
\begin{align}
\delta S =\delta \int_{t_1}^{t_2}L(q(t),\dot{q}(t))\,dt = 0,\label{Cont-var-prin-general}
\end{align}
where the dimension of $q$-space is $m$.
Assume for now that the Lagrangian $L$ is \emph{non-degenerate}, which means the velocity space Hessian
\begin{align}
M_{ij}(q,\dot{q}) = \frac{\partial^2L}{\partial \dot{q}^i\,\partial \dot{q}^j}
\end{align}
is invertible for each $(q,\dot{q})$. The most important consequence of non-degeneracy is that it implies the Euler-Lagrange equations
\begin{align}
\frac{d}{dt}\frac{\partial L}{\partial\dot{q}^i} = \frac{\partial L}{\partial q^i}\label{cts_ele}
\end{align}
are a system of $m$ second-order ordinary differential equations on $q$-space. Indeed, because we have
\begin{align}
\frac{d}{dt}\frac{\partial L}{\partial\dot{q}^i} = M_{ij}\ddot{q}^j + \frac{\partial L}{\partial \dot{q}^i\partial q^j} \dot{q}^j
\end{align}
(we use the standard summation convention for repeated indices), Eq.\,\eqref{cts_ele} may be rewritten as
\begin{align}
\ddot{q}^l = (M^{-1})^{lj}\bigg(\frac{\partial L}{\partial q^j} - \frac{\partial^2L}{\partial\dot{q}^j\partial q^i}\dot{q}^i\bigg),
\end{align}
which for a nondegenerate Hessian is a system of $m$ second-order ordinary differential equations for $q$. Per the usual prescription, such equations are also equivalent to system of $2m$ first-order ordinary differential equations advancing $(q,\dot{q})$. Therefore $(q,\dot{q})$-space is a suitable phase space for such a non-degenerate Lagrangian system. Another important property of non-degenerate Lagrangian systems is that it is possible, at least locally, to perform a Legendre transformation to obtain a Hamiltonian system in canonical variables, $q^i$, $p_i=\partial L/\partial q^i$.

According to Ref.\,\onlinecite{Marsden_2001}, a variational integrator for a system with Lagrangian $L$ is a time-marching algorithm that may be derived from the discrete-time approximation of Eq.\,(\ref{Cont-var-prin-general}), leading to a   discrete variational principle based on the action $S_d$,
\begin{align}
\delta S_d =\delta \sum_{k = N_1}^{N_2} h L_d(q_k,q_{k+1}) = 0,\label{discrete_vp_general}
\end{align}
where the \emph{discrete Lagrangian} $L_d$ is chosen as a specific approximation to the windowed time average of the Lagrangian according to
\begin{align}\label{av_ac}
L_d(q_k,q_{k+1}) \approx \frac{1}{h}\int_{t_k}^{t_{k+1}} L(q(t),\dot{q}(t))\,dt.
\end{align}
Here $h$ denotes the spacing of a uniform temporal grid $t_k=k h$. The quantities $q_k,q_{k+1}$ are evaluated at the endpoints of the interval $[t_k,t_{k+1}]$. The \emph{discrete Euler-Lagrange} (DEL) equations associated with the variational principle \eqref{discrete_vp_general} are given by 
\begin{align}
\frac{\partial L_d}{\partial q^i_2}(q_{k-1},q_{k}) + \frac{\partial L_d}{\partial q^i_1}(q_{k},q_{k+1}) = 0.\label{del_general}
\end{align}
Provided the \emph{discrete Hessian}
\begin{align}
\mathcal{M}_{ij}(q_0,q_1) = \frac{\partial L_d}{\partial q_0^i\partial q_1^j }(q_0,q_1)
\end{align} 
is invertible for each $(q_0,q_1)$, the discrete Euler-Lagrange equations define a mapping
\begin{align}
(q_k,q_{k+1})\mapsto (q_{k+1},q_{k+2}^*(q_k,q_{k+1})),\label{nd_one_step}
\end{align}
where $q_{k+2}^*(q_k,q_{k+1})$ is the solution of Eq.\,\eqref{del_general} for $q_{k+1}$ as a function of $(q_{k-1},q_{k})$ guaranteed by the implicit function theorem. 

When the continuous-time Lagrangian $L$ is non-degenerate, i.e.~the discrete Hessian is invertible, repeated application of the mapping \eqref{nd_one_step} defined by the discrete Euler-Lagrange equations \eqref{del_general} generates a sequence $k\mapsto q_k$ that approximates a solution of Eq.\,\eqref{cts_ele} sampled at the times $t_k = hk$. In particular note that such a discrete-time trajectory requires $2m$ initial conditions to supply to \eqref{nd_one_step}, which is the \emph{same} number of initial conditions required to specify a solution of the continuous-time Euler-Lagrange equations \eqref{cts_ele}.

Suppose now that the continuous-time Lagrangian $L$ has the general form
\begin{align}
L(z,\dot{z}) = \vartheta_i(z)\dot{z}^i - H(z).\label{General-Lagrangian}
\end{align}
Lagrangians of this form in arbitrary variables are known as \emph{phase space Lagrangians},\cite{Cary_1983} and are notable because every Hamiltonian system on an exact symplectic manifold is governed by such a Lagrangian. The  class of Lagrangians in Eq.~\eqref{General-Lagrangian} is a generalization of phase space Lagrangians of the form $p_i \dot{q}^i-H(q,p)$, in canonical variables.  Here we have conformed with standard conventions when discussing phase space Lagrangians in noncanonical variables by making the notational change $q\mapsto z$. 

Because $L$ is linear in the velocities, the velocity-space Hessian is zero, $M_{ij} = 0.$
The Euler-Lagrange equations (\eqref{cts_ele} with $q\rightarrow z$)  therefore cannot be equivalent to a system of $2m$ first-order ordinary differential equations on $(z,\dot{z})$-space, or equivalently a system of second-order equations on $m$-dimensional $z$-space. Instead they are equivalent to the \emph{first-order} system on the $m$-dimensional $z$-space given by
\begin{align}
\dot{z}^i\omega_{ij} = \frac{\partial H}{\partial z^j},\label{psl_ele}
\end{align}
where the $m\times m$ antisymmetric matrix $\omega_{ij}$ is defined by
\begin{align}
\omega_{ij}(z) = \frac{\partial \vartheta_i}{\partial z^j} - \frac{\partial \vartheta_j}{\partial z^i}.\label{omega-ij}
\end{align}
If the functions $\vartheta_i$ have the property that $\omega_{ij}$ is invertible for each $z$, this yields a system of $m$ \emph{first order} equations. It follows that the phase space for a phase space Lagrangian has dimension $m$, which is \emph{half} the dimension of the phase space for a non-degenerate Lagrangian. In particular, the number of initial conditions required to specify a solution of Eq.\,\eqref{psl_ele} is $m$ instead of $2m$.

Application of variational integration to a phase space Lagrangian is usually problematic for the following reason. Because invertible matrices are generic,
most choices of discrete Lagrangian $L_d$ will have an invertible discrete Hessian $\mathcal{M}$.
One consequence of the invertibility of the discrete Hessian, which is suggestive that something is wrong, is that the discrete Euler-Lagrange equations for $L_d$ involve three time levels, as in Eq.\,\eqref{del_general},  and therefore require $2m$ initial conditions to generate the discrete-time trajectory $k\mapsto z_k$. Indeed, from the perspective of the discrete Euler-Lagrange equations \eqref{del_general}, a discretized phase space Lagrangian is no different from a discrete Lagrangian coming from a non-degenerate continuous-time Lagrangian. Therefore the mapping \eqref{nd_one_step} is still well-defined, which implies that the discrete system derived from a generic discrete Lagrangian requires $2m$ initial conditions, in spite of the fact that the underlying continuous-time system \eqref{psl_ele} requires only $m$ initial conditions.

The preceding argument shows that typical variational integrators for phase space Lagrangians are \emph{multi-step methods}. Multi-step methods generally have \emph{parasitic modes}, which may be unstable. Nevertheless, many multi-step methods have favorable numerical performance, in spite of the existence of these parasitic modes.\cite{Hairer_2006} In a reasonable scheme for a dissipative system, such parasitic modes damp out harmlessly in the early stages of the integration. However, the parasitic modes of multistep variational integrators typically do not damp out. To understand why, we note that multi-step integrators arising from discrete phase space Lagrangians must have either neutrally stable parasitic modes or have one growing parasitic mode for each damped parasitic mode.\cite{Ellison_2018} This symmetry between damped and growing parasitic modes is a consequence of the preservation of a symplectic form on the $2m$-dimensional space of pairs $(z_1,z_2)$. Thus, the parasitic modes associated with a discrete phase space Lagrangian may be neutrally stable at best, but still susceptible to driving by nonlinear terms. Unfavorable behavior of parasitic modes arising from examples of discrete phase space Lagrangians are described in Refs.\,\onlinecite{Ellison_2018} and \onlinecite{Ellison_thesis}.

In Ref.\,\onlinecite{Ellison_2018}, Ellison \,\emph{et.\,al.} observed that the following special class of discrete phase Lagrangians avoid these multistep issues in discretizing
phase space Lagrangians.

\begin{definition}[properly-degenerate discrete Lagrangian]
Suppose the dimension $m$ of phase space (coordinates $z$) is even. A discrete Lagrangian $L_d(z_0,z_1)$ is \emph{properly-degenerate} if the rank of the discrete Hessian is everywhere half-maximum. In other words the rank of $\mathcal{M}_{ij}(z_0,z_1)$ is $m$ rather than $2m$ for all $(z_0,z_1)$.
\end{definition}

\begin{remark}
We write the dimension as $2m$; this dimension is necessarily even when working with phase space Lagrangians because an antisymmetric matrix like $\omega_{ij}$ always has a non-trivial null space in odd dimensions. So Eq.\,\eqref{psl_ele} cannot be solved
for $\dot{z}^i$.
\end{remark}


A simple example of a non-properly degenerate discrete Lagrangian in a  one degree of freedom ($m=1$) system in canonical variables follows from a centered discretization of $S=\int(p\dot{q}-H)dt$,
\begin{equation}
L_{d}\left(q_{0},p_{0},q_{1},p_{1}\right)=\frac{p_{0}+p_{1}}{2}\frac{q_{1}-q_{0}}{h}-H\left(\frac{q_{0}+q_{1}}{2},\frac{p_{0}+p_{1}}{2}\right).\label{eq:Ld-explicit-midpoint}
\end{equation}
The discrete Euler-Lagrange (DEL) equations are
\[
\frac{\partial L_{d}\left(q_{0},p_{0},q_{1},p_{1}\right)}{\partial q_{1}}+\frac{\partial L_{d}\left(q_{1},p_{1},q_{2},p_{2}\right)}{\partial q_{1}}=0
\]
and
\[
\frac{\partial L_{d}\left(q_{0},p_{0},q_{1},p_{1}\right)}{\partial p_{1}}+\frac{\partial L_{d}\left(q_{1},p_{1},q_{2},p_{2}\right)}{\partial p_{1}}=0,
\]
leading to
\[
\frac{p_{2}-p_{0}}{2h}=-\frac{1}{2}H_{1}\left(0,1\right)-\frac{1}{2}H_{1}\left(1,2\right),
\]\[
\frac{q_{2}-q_{0}}{2h}=\frac{1}{2}H_{2}\left(0,1\right)+\frac{1}{2}H_{2}\left(1,2\right),
\]
where $H_{1}$ and $H_{2}$ are derivatives of $H$ with respect to
its first and second arguments, respectively, and the symbols $(0,1),(1,2)$
represent $(q_{0}+q_{1})/2,(p_{0}+p_{1})/2$ and  $(q_{1}+q_{2})/2,(p_{1}+p_{2})/2$, respectively.
This is clearly a multistep method, linking steps $k=0$, $k=1$,
$k=2$, i.e.~giving second order difference equations for $q_{k}$
and $p_{k}$. (The finite difference forms for $\dot{q}$ and $\dot{p}$ suffer from ``stencil spreading.") The discrete Hessian for this system is
\[
\left(\begin{array}{cc}
\partial^{2}L_{d}/\partial q_{0}\partial q_{1} & \partial^{2}L_{d}/\partial q_{0}\partial p_{1}\\
\partial^{2}L_{d}/\partial p_{0}\partial q_{1} & \partial^{2}L_{d}/\partial p_{0}\partial p_{1}
\end{array}\right)
\]
\[
=\left(\begin{array}{cc}
-\frac{H_{11}\left((q_{0}+q_{1})/2,(p_{0}+p_{1})/2\right)}{4} & \,\,\,-\frac{1}{2h}-\frac{H_{12}\left(q_{0}+q_{1})/2,(p_{0}+p_{1})/2\right)}{4}\\
\frac{1}{2h}-\frac{H_{21}\left(q_{0}+q_{1})/2,(p_{0}+p_{1})/2\right)}{4} & -\frac{H_{22}\left(q_{0}+q_{1})/2,(p_{0}+p_{1})/2\right)}{4}
\end{array}\right),
\]
generically of full rank, showing agreement between the multistep
property of the DEL equations and the rank of the discrete Hessian for this system. And, indeed, this system requires an extra set of initial conditions and exhibits parasitic modes.

This example is to be contrasted with the first order accurate case arising from the discrete phase space Lagrangian, again in one degree of freedom and in canonical variables,
\begin{align}
L_{d}=p_{k}\frac{q_{k+1}-q_{k}}{h}-H(q_{k+1},p_{k}),\label{eq_SE-variational_P}
\end{align}
leading to
\begin{align}
q_{k+1}=q_{k}+hH_{2}(q_{k+1},p_{k}),\,\,\,p_{k+1}=p_{k}-hH_{1}(q_{k+1},p_{k}),\label{SEScheme}
\end{align}
a form of the symplectic Euler scheme\cite{Hairer_2006}, with $q$ updated
implicitly and used in a leapfrog manner in the explicit update of $p$. 
The discrete Hessian is
\[
\left(\begin{array}{cc}
0 & \,\,\,0\\
\frac{1}{h}-H_{21}\left(q_{1},p_{0}\right) & \,\,\,0
\end{array}\right),
\]
of rank $m=1$, traced to the fact that $L_d$ in Eq.~\eqref{eq_SE-variational_P} depends on $p$ at only one step. This result is consistent with the single-step nature of the scheme. The adjoint scheme $S_h^{\dagger}$ also has a discrete Hessian with rank $m=1$.

Reference \onlinecite{Ellison_thesis} proves that properly-degenerate discrete phase space Lagrangians are necessarily free of parasitic modes. This result suggests, but does not directly imply, that variational integrators derived from properly-degenerate discrete phase space Lagrangians are single-step methods instead of multi-step methods. In fact, under mild technical hypotheses, properly-degenerate discrete phase space Lagrangians are indeed single-step methods. 

The simplest way to understand the single-step nature of variational integrators obtained from properly-degenerate discrete phase space Lagrangians is to restrict our attention to the linearized discrete Euler-Lagrange equations. 

\begin{theorem}[Linearized single-step property]\label{l_onestep}
Let $L_d(z_1,z_2)$ be a properly-degenerate discrete Lagrangian satisfying the mild technical hypotheses (G1)-(G2) described in Appendix A.
Then the discrete Euler-Lagrange equations linearized about a trajectory $k\mapsto z^0_k$ are equivalent to a single-step method.
\end{theorem}

\begin{remark}
For a complete statement and proof of Theorem \ref{l_onestep}, see Theorem \ref{l_onestep_appendix} in Appendix A.
\end{remark}

\begin{proof}[Proof sketch]
Let $k\mapsto z_{k}^\epsilon$ (the limit $\epsilon\rightarrow 0$ represents a reference trajectory), be a smooth $\epsilon$-dependent family of solutions of the discrete Euler-Lagrange equations associated with a properly-degenerate discrete Lagrangian,
\begin{align}
\frac{\partial L_d}{\partial z_2^i}(z_{k-1}^\epsilon,z_k^\epsilon) + \frac{\partial L_d}{\partial z_1^i}(z_{k}^\epsilon,z_{k+1}^\epsilon) = 0,
\end{align} 
where $\partial/\partial z_{1}$ and $\partial/\partial  z{_2}$ refer to derivatives with respect to the first and second arguments of $L_d$. Differentiating the discrete Euler-Lagrange equations with respect to $\epsilon$ at $\epsilon=0$ shows that the linearization $k\mapsto \delta z_k = \partial_\epsilon (z_k^\epsilon)_{\epsilon=0}$ of a trajectory near $k\mapsto z_k^0$ satisfies the linearized discrete Euler-Lagrange equations
\begin{align}
A_{ij}(k)\delta z_{k-1}^j + C_{ij}(k)\delta z_{k}^j + B_{ij}(k)\delta z_{k+1}^j = 0,
\end{align}
where we have introduced the convenient shorthand notation
\begin{align}
A_{ij}(k) =& \mathcal{M}_{ji}(z_{k-1}^0,z_k^0)\label{A_shorthand}\\
B_{ij}(k) = & \mathcal{M}_{ij}(z_k^0,z_{k+1}^0)\label{B_shorthand}\\
C_{ij}(k) = & \frac{\partial L_d}{\partial^2 z_2^i\partial z_2^j}(z_{k-1}^0,z_{k}^0) + \frac{\partial^2 L_d}{\partial z_1^i \partial z_1^j}(z_k^0,z_{k+1}^0).\label{C_shorthand}
\end{align}
Note that $C_{ij}(k) = C_{ji}(k)$ is a symmetric matrix, while $A_{ij}(k+1) = B_{ji}(k)$ are transposes of one another after a time step shift. 

Let $[Z]$ denote the $m\times m$-matrix whose components are $Z_{ij}$. Set $X_k = \text{im}\, [A(k)]$ and $Y_k=\text{im}\, [B(k)]$, where $\text{im}$ denotes the range/column space of the matrix.
By proper degeneracy and hypothesis (G1), $\text{dim} X_k = \text{dim} Y_k = m/2$ and $X_k\cap Y_k = \{0\}$. Therefore $\mathbb{R}^m= X_k\oplus Y_k$ decomposes as a direct sum for each $k$. Associated with this direct sum is the pair of projection matrices $[\pi_X(k)]:\mathbb{R}^m\rightarrow X_k$ and $[\pi_Y(k)]:\mathbb{R}^m\rightarrow Y_k$. Applying the projection $[\pi_X(k)]$ to the linearized discrete Euler-Lagrange equations gives
\begin{align}
A_{ij}(k)\delta z_{k-1}^j + (\pi_X(k))_{i\overline{i}}C_{\overline{i}j}(k)\delta z_{k}^j =0,\label{X_proj}
\end{align}
while applying the projection $[\pi_Y(k)]$ gives
\begin{align}
(\pi_Y(k))_{i\overline{i}}C_{\overline{i}j}(k)\delta z_{k}^j + B_{ij}(k)\delta z_{k+1}^j = 0.\label{Y_proj}
\end{align}
In particular by shifting \eqref{X_proj} ahead by one timestep we obtain the implicit linear relation between $\delta z_k$ and $\delta z_{k+1}$ given by
\begin{align}
0=&A_{ij}(k+1)\delta z_{k}^j + (\pi_X(k+1))_{i\overline{i}}C_{\overline{i}j}(k+1)\delta z_{k+1}^j\label{lprf_one}\\
0=& (\pi_Y(k))_{i\overline{i}}C_{\overline{i}j}(k)\delta z_{k}^j + B_{ij}(k)\delta z_{k+1}^j.\label{lprf_two}
\end{align}
To complete the proof it is enough to demonstrate that for each $\delta z_k$, there exists a unique $\delta z_{k+1}$ that satisfies Eqs.\,\eqref{lprf_one}-\eqref{lprf_two}. This is done in the proof of Theorem \ref{l_onestep_appendix} in Appendix A.
\end{proof}

\begin{remark}
Theorem \ref{nl_onestep_appendix} in Appendix A uses the above result and the implicit function theorem to prove that DVIs are also one-step methods at the nonlinear level.
\end{remark}

We therefore have the following simple explanation for the absence of parasitic modes in variational integrators derived from properly-degenerate discrete phase space Lagrangians. Because parasitic modes only arise in multi-step schemes, and DVIs are equivalent to single-step schemes by Theorem \ref{nl_onestep_appendix}, parasitic modes are not generated by DVIs. Note that Ref.\,\onlinecite{Ellison_thesis} proves the absence of parasitic modes using less direct arguments, but \emph{does not} prove that DVIs are generally $1$-step methods. (The $1$-step property was observed in examples, however.) Theorems \ref{l_onestep} and \ref{nl_onestep_appendix} therefore give a more detailed understanding of the benefits of DVI.

\section{Second-order DVI\label{sec:DVI2}}
We now turn to the task of constructing degenerate variational integrators with second-order accuracy. As in Refs.\,\onlinecite{Ellison_thesis} and  \onlinecite{Ellison_2018}, we focus on noncanonical phase space Lagrangians of the form
\begin{align}\label{usual_suspect}
L(x,y,\dot{x},\dot{y}) = f_i(x,y)\,\dot{x}^i - H(x,y),
\end{align}
where the dimension $m$ is even and $i=1,\dots,m/2$.  Such a phase space Lagrangian is a special case of the general phase space Lagrangian in noncanonical variables in Eq.\,(\ref{General-Lagrangian}), without terms proportional to $\dot{y}_i$. This form is sufficient to cover the important examples of magnetic field line flow and guiding center dynamics tested in this paper. In toroidal geometries of interest to magnetic fusion, the general guiding center Lagrangian may be placed in the form \eqref{usual_suspect} through the use of toroidal regularization,\cite{Burby_Ellison_2017}. Also, as we will discuss, the field-line Lagrangian may always be brought into the form \eqref{usual_suspect} using a simple gauge transformation.


\subsection{Composing a first order scheme with its adjoint?}\label{Sec.Composing}

A common method for constructing a second-order accurate integrator starting from a first-order accurate integrator is to compose it with its adjoint\cite{Hairer_2006}, discussed above for the special case of the symplectic Euler scheme. To show what can go wrong with composing a first order scheme with
its adjoint in the context of variational integration, let us introduce the discrete phase space Lagrangian
for a canonical system in one degree of freedom,
\begin{equation}
L_{d}=p_{k}\frac{q_{k+1}-q_{k}}{h}-H(q_{k},p_{k}).\label{eq:Mod-SE-Ld}
\end{equation}
Except for the change $H(q_{k+1},p_{k})\to H(q_{k},p_{k})$, this discretization is identical to that
of the symplectic Euler scheme of Eqs.\,\eqref{eq_SE-variational_P} and \eqref{SEScheme}. Variations with respect
to $p_{k}$ and $q_{k}$ yield the map
\begin{equation}
q_{k+1}=q_{k}+hH_{2}(q_{k},p_{k}),\,\,\,p_{k+1}=p_{k}-hH_{1}(q_{k+1},p_{k+1}).\label{eq:Mod-SE-scheme}
\end{equation}
The first of these is explicit with respect to $q$; the second update is implicit in $p$ \emph{and} leapfrogged
in $q$, so that this scheme is slightly different from the updating in symplectic Euler.
The adjoint of this scheme follows from $k\leftrightarrow k+1,\,\,h\rightarrow-h$. We find
\begin{equation}
L_{d}=p_{k+1}\frac{(q_{k+1}-q_{k})}{h}-H(q_{k+1},p_{k+1}),\label{eq:Adj-Mod-SE-Ld}
\end{equation}
leading to
\begin{equation}
p_{k+1}=p_{k}-hH_{1}(q_{k},p_{k}),\,\,\,q_{k+1}=q_{k}+hH_{2}(q_{k+1},p_{k+1}).\label{eq:Adj-Mod-SE-scheme}
\end{equation}
This scheme is explicit in $p$, implicit and leapfrogged in $q$. 

By direct substitution,
we find that
\begin{equation}
\left(1+hH_{12}(q,p)\right)dq\wedge dp\label{eq:Two-form-with-O(h)}
\end{equation}
is preserved by this scheme. As discussed in Ref.\,\onlinecite{Ellison_2018}, the preservation of this two-form can also be shown by
inspecting the first and last terms in $dS=d\sum_{k}hL_{d}$.
This property is consistent with the fact that this scheme and the adjoint symplectic Euler scheme are equivalent under a noncanonical change of variables. Unlike the symplectic Euler scheme, which preserves the canonical
two-form $\omega=dq\wedge dp$, this form has $\omega=\omega(h)=(1+O(h))dq\wedge dp$. 
That is, in this scheme the phase space coordinates $(q,p)$ are not
canonical. By analogous arguments, or by direct inspection, the adjoint
scheme preserves
\begin{equation}
\omega^{\dagger}(h)=\omega(-h)=\left(1-hH_{12}(q,p)\right)dq\wedge dp\label{eq:Two-form-adjoint},
\end{equation}
and because the correction is $O(h)$, the two-forms differ by $O(h)$.
From these observations, when $\omega(h)\neq \omega(-h)$ it is not clear how to determine which two-form, if any,
is preserved by the composition of the scheme and its adjoint. In Sec.~\ref{Numerical-no-go-composing} we show numerical evidence that such a composed form does not in general preserve any two-form. Also, in Appendix C we show the direct analogy to the symplectic Euler scheme in Eq.~\eqref{SEScheme} (not  Eq.~\eqref{eq:Mod-SE-scheme}), and show that its preserved two-form also has $O(h)$ corrections. While we cannot rule out the existence of a first order accurate scheme in our class of noncanonical variables (which would allow composition with its adjoint to obtain second order accuracy), we proceed in the next section to show two separate second order accurate schemes for such systems.

\subsection{Centered schemes for second order accuracy}

Because composing a first-order DVI with its adjoint does not reliably produce a scheme preserving a symplectic form (although such a scheme is second-order accurate), we are naturally led to consider the problem of proceeding to higher order by identifying improved properly degenerate discrete Lagrangians.
We first illustrate for a one degree of freedom case ($m=2$) in canonical variables,  introducing
two different staggered, centered schemes to discretize the action
$S=\int Ldt$ for the phase space Lagrangian $L(q,p,\dot{q},\dot{p})=p\dot{q}-H(q
,p)$.
The first scheme has
\[
L_{d}(q_{k},q_{k+1},p_{k+1/2})=p_{k+1/2}\frac{(q_{k+1}-q_{k})}{h}-H\left(\frac{q
_{k}+q_{k+1}}{2},p_{k+1/2}\right).
\]
This scheme includes a staggered nature of $p$ and a midpoint nature with respect
to $q$. 
Taking variations with respect to $q_{k}$ and $p_{k+1/2}$ we find

\begin{align}
q_{k+1} & =q_k+h H_2\left(\frac{q_k+q_{k+1}}{2},p_{k+1/2}\right),\\
p_{k+1/2}& =p_{k-1/2}-\frac{h}{2} H_1\left(\frac{q_k+q_{k+1}}{2},p_{k+1/2}\right)
-\frac{h}{2} H_1\left(\frac{q_{k-1}+q_{k}}{2},p_{k-1/2}\right)
\end{align}
We call this the \emph{midpoint DVI} (MDVI) scheme. This scheme is clearly time-centered, which suggests second-order accuracy, and must be advanced implicitly in both variables. It can easily be shown to be properly degenerate, by showing that the rank of the $2\times2$ discrete Hessian is one, essentially because, as for the symplectic Euler scheme, its adjoint, and the schemes of Eqs.~\eqref{eq:Mod-SE-Ld} and \eqref{eq:Adj-Mod-SE-Ld}, the discrete Lagrangian $L_d$ depends on $p$ at only one time level.
This degeneracy is in spite of the fact that the scheme appears to be a two-step scheme, connecting $q_{k-1}$, $q_{k}$ and $q_{k+1}$ (but only $p_{k-1/2}$ and $p_{k+1/2}$). As in Ref.\,\onlinecite{Ellison_2018}, $q_{k-1}$ can be expressed in terms of $q_{k}$ and $p_{k-1/2}$; i.e., it defines a time-advance map of the form $(q_k, p_{k-1/2}) \mapsto (q_{k+1}, p_{k+1/2})$, namely, a one-step method. This property of appearing to be two-step but showing a single-step nature after some substitutions,\cite{Ellison_2018} shows the importance of the discrete Hessian test; indeed, without the assurance of the discrete Hessian test, it would be easy to miss the possibility of this substitution. 

The second scheme uses
\[
L_{d}(q_{k},q_{k+1},p_{k+1/2})=p_{k+1/2}\frac{(q_{k+1}-q_{k})}{h}-
\frac{1}{2}H\left(q_{k},p_{k+1/2}\right)-\frac{1}{2}H\left(q_{k+1},p_{k+1/2}
\right),
\]
and variations lead to
\begin{align}
q_{k+1}=q_k+\frac{h}{2}H_2(q_k,p_{k+1/2})+\frac{h}{2}H_2(q_{k+1},p_{k+1/2}),\\
p_{k + 1/2}=p_{k - 1/2}-\frac{h}{2} H_1\left(q_{k},p_{k + 1/2}\right)
-\frac{h}{2} H_1\left(q_{k},p_{k-1/2}\right).
\end{align}
This second scheme is also centered, again suggesting second order accuracy, and also must be advance implicitly in both variables. For this scheme, the discrete Lagrangian is obtained using a trapezoidal quadrature scheme, hence we call this scheme the \emph{trapezoidal DVI} (TDVI) scheme. In this scheme, the single-step nature is evident from the discrete Euler-Lagrange equations; of course, the Hessian test confirms the single-step character. Finally, a backward error analysis for either the MDVI or the TDVI scheme, with $q_{k}=q(t=kh)$ and $p_{k+1/2}=p((k+1/2)h)$, shows second-order accuracy. 

In both schemes, if initial conditions $q_{0}$ and $p_{0}$ are both given 
at $t=0$, the momentum variable needs to be regressed to $p_{-1/2}$
by a processing scheme, which we will discuss shortly.
Numerical trials using the non-reversible (c.f.~Appendix B) Hamiltonian $H=(p^{2}+q^{2})/2+\alpha q p^{3}/3$ indicate second order accuracy and the good long-time behavior of
a scheme with a preserved two-form.



For the noncanonical (but not completely general) phase space Lagrangians of the form of Eq.\,\eqref{usual_suspect}, the midpoint DVI scheme has
\begin{equation}
    L_d(x_k,y_{k + 1/2}, x_{k+1}) = f_i\left( \frac{x_k+x_{k+1}}{2}, y_{k+1/2}\right ) \frac{x^i_{k+1}-x^i_k}{h}  - H\left ( \frac{x_k+x_{k+1}}{2},y_{k+1/2}\right).\label{centered_LD}
\end{equation}
Again, the centeredness suggests, and backward error analysis indeed shows, second order accuracy. Also, because the staggered-grid discrete Lagrangian \eqref{centered_LD} is merely a relabeling the of a first-order non-staggered-grid discrete Lagrangian introduced in Ref.\,\onlinecite{Ellison_thesis}, the DVI associated with \eqref{centered_LD} automatically preserves a symplectic form. Again, the discrete Hessian has half rank because $L_d$ depends on $y$ at only one time level.

The discrete Euler-Lagrange equations stemming from Eq.~\eqref{centered_LD} are given by
\begin{subequations}
\begin{align}
    \half \left( f_{i,j}\left(k + \half\right)(x_{k+1}^i - x_{k}^i)  + f_{i,j}\left(k-\half\right) (x_k^i - x_{k-1}^i)  \right) - \quad & \nonumber \\
    \left(f_j\left(k +  \half\right) - f_j\left(k - \half\right)\right) - \frac{h}{2} \left(H_{,j}\left(k + \half\right) + H_{,j}\left(k-\half \right) \right) & = 0 \\
    f_{i,\alpha}\left(k + \half \right) \left(x_{k+1}^i - x_k^i\right) - h H_{,\alpha}\left(k + \half\right) & = 0,
\end{align}
\label{mdvi-noncanonical}
\end{subequations}
where $i,j \in \{1, ..., m/2\}$, $m$ is even, while $\alpha \in \{m/2+1, ... , m\}$ and $(k+1/2)$ is shorthand for the arguments $\left(\frac{x_{k+1} + x_k}{2}, y_{k+1/2} \right)$ and $(k-1/2)$ is analogous. Like the canonical MDVI, the apparent dependence on $x_{k-1}$ may be eliminated in favor of a function depending on $(x_k, y_{k-1/2})$. In practice one has access to $x_{k-1}$ due to the state at prior iterations (except for on the first step) and the stored value may be used instead of directly solving for $x_{k-1}$ at each new iteration.

The trapezoidal DVI scheme for this class of noncanonical cases arises from
\begin{equation}
\begin{split}\label{centered_LD-prime}
L_d(x_k,y_{k+1/2},x_{k+1}) =& \tfrac{1}{2}\left(f_i(x_k, y_{k+1/2})+f_i(x_{k+1}, y_{k+1/2})\right) \frac{x^i_{k+1}-x^i_k}{h} \\
&- \tfrac{1}{2}\left(H(x_{k},y_{k+1/2}) + H(x_{k+1},y_{k+1/2}) \right).
\end{split}
\end{equation}
Again, the time-centered property leads to second order accuracy, and its discrete Hessian shows that it is properly degenerate, again because $L_d$ depends on $y$ at only one time level. 

Performing variations with respect to $x_{k}, y_{k+1/2}$ yields the TDVI scheme:
\begin{subequations}
\begin{align}
    \half \left(f_{i,j}(x_k, y_{k+1/2}) \left( x_{k+1}^i - x_{k}^i \right) +  f_{i,j}(x_k, y_{k-1/2}) \left( x_{k}^i - x_{k-1}^i \right) \right) -  &\nonumber \\ 
    \half \left(f_j(x_{k+1}, y_{k+1/2}) + f_{j}(x_{k}, y_{k+1/2}) - f_j(x_k, y_{k-1/2}) - f_j(x_{k-1}, y_{k-1/2} \right) & - \nonumber \\
    \half \left( H_{,j}(x_k, y_{k+1/2}) + H_{,j}(x_{k+1}, y_{k+1/2}) \right) & = 0  \\
    \half \left( f_{i,\alpha} (x_k, y_{k+1/2}) + f_{i,\alpha}(x_{k+1}, y_{k+1/2}) \right) \left(x_{k+1}^i - x_{k}^i\right) - & \nonumber \\
    \tfrac{h}{2} \left( H_{,\alpha}(x_k, y_{k+1/2}) + H_{,\alpha}(x_{k+1}, y_{k+1/2}) \right) &= 0. 
\end{align}\label{tdvi-noncanonical}
\end{subequations}
In contrast to the canonical setting, the TDVI scheme introduces dependence on $x_{k-1}$ in general. Of course, this dependence is superficial, and can be eliminated in favor of a function of $(x_k, y_{k-1/2})$ as previously discussed. Relative to the MDVI scheme, the TDVI scheme requires more function evaluations (i.e., additional evaluations of $f$ and $H$) in the update rule; this has the potential to increase the computational expense of TDVI relative to MDVI.

Of course, initial conditions for the time-marching schemed defined by \eqref{centered_LD} or \eqref{centered_LD-prime} will be supplied at an integer timestep instead of directly on the staggered grid. In order to transform integer timestep initial conditions $(x_0,y_0)$ to staggered-grid initial conditions $(x_0,y_{-1/2})$, it is sufficient to advance $y$ backward in time by a half-step $h/2$ using any first (or higher) order accurate scheme. Encode this transformation in the mapping $\varphi_{-h/2}:(x_0,y_0)\mapsto (x_0,y_{-1/2})$. Similarly, at the end of a simulation, the staggered grid data $(x_N,y_{N-1/2})$ must be collocated to the integer grid data $(x_N,y_N)$. The natural way to do this is simply to apply the inverse of $\varphi_{-h/2}$, i.e. set $(x_N,y_N) = \varphi_{-h/2}^{-1}(x_N,y_{N-1/2})$, and second order accuracy is preserved. These two conditions are special cases of enforcing the relationship 
\begin{align}
\varphi_{-h/2}(x_k,y_{k}) = (x_k,y_{k-1/2})\label{CollocationScheme}
\end{align}
for all $k$. Indeed, for displaying results during a computation, e.g.~where $k$ equals a multiple of a fundamental period $M$, typically with $M>1$, these results will retain second order accuracy if the points are collocated according to Eq.~\eqref{CollocationScheme}. In Ref.\,\onlinecite{Blanes_2004}, the idea of increasing the order of a low-order scheme using a map and its inverse as pre- and post-processors is explored in greater detail. It would be interesting to determine if even higher-order DVIs may be derived using more elaborate processing than the more-or-less obvious processors described here.


\section{Magnetic field line and guiding center examples\label{MFL-GC}}

In this section we apply the discretizations of Eqs.\,\eqref{centered_LD} and \eqref{centered_LD-prime} to the Lagrangians for the magnetic field line problem and the guiding center system, both described in Refs.\,\onlinecite{Ellison_2018,Ellison_thesis}.

\subsection{Magnetic field line\label{MFL}}
For the problem of tracing magnetic field lines, we take the action to be equal to the flux
\begin{align}
    \Phi= \int L\left(\mathbf{x}, \frac{d \mathbf{x}}{d t}\right) dt = \int \mathbf{A}(\mathbf{x}) \cdot \frac {d\mathbf{x}}{dt}dt,\label{Adot-dx}
\end{align}
where $\mathbf{A}$ is the magnetic vector potential. The invariance with respect to reparameterization of time\cite{Ellison_2018} in Eq.~\eqref{Adot-dx} is consistent with its Euler-Lagrange equation, namely $(d\mathbf{x}/dt)\times \mathbf{B}=0$, 
where $\mathbf{B}=\nabla \times \mathbf{A}$ is the magnetic field: this equation determines the direction of the flow but not the speed. This time invariance is dealt with by parameterizing the field line trajectory in terms of one of the coordinates (without loss of generality, we choose $x^3$) instead of ``time" $t$. 
\begin{align}
    \Phi=\int L(x^1,x^2,x^3,dx^2/dx^3,x^3)dx^3 
    =\int \left( A_2(x^1,x^2,x^3)dx^2/dx^3+A_3(x^1,x^2,x^3)\right)dx^3.\label{MFL-z-action}
\end{align}
This Lagrangian is of the form in Eq.\,\eqref{usual_suspect}, with $f = [0 \quad A_2]$  and $H=-A_3$; the Euler-Lagrange equations for Eq.\,\eqref{MFL-z-action} are
\begin{subequations}
    \begin{align}
    dx^1/dx^3 &=(A_{3,2}-A_{2,3})/A_{2,1} = B^1/B^3 \\
    dx^2/dx^3 &=-A_{3,1}/A_{2,1} = B^2/B^3.
    \end{align}
\end{subequations}

The MDVI then follows from Eq.~\eqref{mdvi-noncanonical}, which when expressed in terms of the magnetic vector potential becomes:
\begin{subequations}
\begin{align}
A_{2,1}(k + \half) \left(x^2_{k+1} - x^2_k \right) + h A_{3,1}(k+\half) & = 0 \\
\half\left(A_{2,2}(k + \half) \left(x^2_{k+1} - x^2_k \right) + A_{2,2}(k - \half) \left(x^2_{k} - x^2_{k-1} \right) \right) - \quad & \nonumber \\
\left( A_2(k+\half) - A_2(k-\half) \right) + \tfrac{h}{2}\left( A_{3,2}(k+\half) + A_{3,2}(k-\half) \right) & = 0,
\end{align}
\end{subequations}
where $(k+\half)$ denotes evaluation at $(x^1_{k+1/2}, (x^2_{k+1} + x_k^2)/2, x^3_{k+1/2})$.

imilarly, the TDVI algorithm for the magnetic field line problem follows from Eq.~\eqref{tdvi-noncanonical}:
\begin{subequations}
   \begin{align}
   &\half \left(A_{2,1}(x_{k+1/2}^1, x_{k+1}^2, x_{k+1/2}^3) + A_{2,1}(x_{k+1/2}^1, x_{k}^2, x_{k+1/2}^3) \right) \left(x_{k+1}^2 - x_k^2 \right) + \nonumber \\
   & \tfrac{h}{2} \left(A_{3,1}(x_{k+1/2}^1, x_{k}^2, x_{k+1/2}^3) + A_{3,1}(x_{k+1/2}^1, x_{k+1}^2, x_{k+1/2}^3) \right) = 0 \\
&  \tfrac{1}{2} \left(A_{2,2}(x^1_{k+1/2}, x^2_k, x^3_{k+1/2})(x^2_{k+1} - x^2_{k}) + A_{2,2}(x^1_{k-1/2}, x^2_k, x^3_{k-1/2})(x^2_k - x^2_{k-1}) \right)  + \nonumber \\
& \tfrac{1}{2} \left( A_2(x^1_{k-1/2}, x^2_{k-1}, x^3_{k-1/2}) + A_2(x^1_{k-1/2}, x^2_{k}, x^3_{k-1/2})\right) - \nonumber \\ 
& \half \left( A_2(x^1_{k+1/2}, x^2_{k}, x^3_{k+1/2}) + A_2(x^1_{k+1/2}, x^2_{k+1}, x^3_{k +1/2})\right) + \nonumber \\
& \frac{h}{2} \left( A_{3,2}(x^1_{k-1/2}, x^2_{k}, x^3_{k-1/2}) + A_{3,2}(x^1_{k+1/2}, x^2_{k}, x^3_{k+1/2}) \right) = 0.
   \end{align}
\end{subequations}

Numerical tests of the midpoint and trapezoidal DVI schemes are presented in Sec.\,\ref{sec:examples}.




\subsection{Guiding center equations\label{GC-formulation}}

We treat the example of the guiding center equations based on discretizing the toroidally-regularized\cite{Burby_Ellison_2017} phase space Lagrangian
\begin{align}
    L(z,\dot{z})&= A_2(\bm{x})\dot{x}^2+(A_3(\bm{x})+u)\dot{x}^3 - H_{\text{gc}} \nonumber\\
    & \equiv A^\dagger_2(\bm{x})\,\dot{x}^2 + A^\dagger_3(\bm{x},u)\,\dot{x}^3 - H_{\text{gc}}
\end{align}
where $H_{gc}=\tfrac{1}{2}(B^3/B)^2\,u^2+\mu B - \tfrac{1}{2}E_\perp^2/B^2 - u\bm{b}\cdot (\bm{E}\times \nabla x^3)/B+\varphi$ is the guiding center Hamiltonian, in toroidally-regularized noncanonical variables,\cite{Littlejohn_1981,Ellison_2018,Burby_Ellison_2017} $u$ is the toroidally-regularized parallel velocity, $\mu$ is the magnetic moment, $B$ is the magnitude of the magnetic field, and $\varphi$ is the scalar potential for the electric field $\bm{E} = -\nabla \varphi $. For simplicity, we neglect time-dependence in the electromagnetic fields. Toroidal regularization requires $|B^3|>0$ everywhere. In toroidal geometries relevant to magnetic fusion energy, possible choices for coordinates include $(x^1,x^2,x^3) = (R,Z,\phi)$ or $(x^1,x^2,x^3) = (r,\theta,\phi)$, where $(R,Z,\phi)$ denote cylindrical coordinates, and $(r,\theta,\phi)$ denote toroidal coordinates. As in the magnetic field example, the gauge condition $A_1=0$ was imposed to lead to proper degeneracy. (The original form \cite{Ellison_2018,Ellison_thesis}, required that the covariant component $B_1$ of the magnetic field also vanish, but in a subsequent improvement,\cite{Burby_Ellison_2017} that condition was relaxed.)

The second-order-accurate DVIs follow by establishing the correspondence with Eq.~\eqref{usual_suspect}; in this case, $f = [0 \quad A_2^\dagger \quad A_3^\dagger \quad 0]$ and $H = H_{gc}$. We discretize this phase space Lagrangian by the midpoint and trapezoidal DVI schemes for second order accuracy. For the midpoint DVI discretization of the guiding center system, we use the discrete Lagrangian
\begin{align}
& L_d(r_{k+1/2}, \theta_{k}, \theta_{k+1}, \phi_k, \phi_{k+1}, u_{k+1/2}) = \nonumber \\
& \quad A_\theta(k+\half) (\theta_{k+1} - \theta_k) +  \left( A_\phi(k+\half) + u_{k+1/2} \right) (\phi_{k+1} - \phi_k) - h H(k+\half).
\end{align}
The MDVI scheme, following from Eq.~\eqref{mdvi-noncanonical}, is
\begin{subequations}
\begin{align}
& A_{\theta, r}(k+\half) (\theta_{k+1} - \theta_k) + A_{\phi, r}(k+\half) (\phi_{k+1} - \phi_k) - h H_{,r}(k+\half) = 0\\
& \tfrac{1}{2} A_{\theta,\theta}(k-\half) (\theta_k - \theta_{k-1}) + \tfrac{1}{2} A_{\theta, \theta}(k+\half) (\theta_{k+1} - \theta_k) + \nonumber \\
& \tfrac{1}{2} A_{\phi,\theta}(k-\half) (\phi_k - \phi_{k-1}) + \tfrac{1}{2} A_{\phi, \theta} (k+\half) (\phi_{k+1} - \phi_k) + \nonumber \\
& A_\theta(k-\half) + u_{k-1/2} - A_\theta(k+\half) - u_{k+1/2} - \tfrac{h}{2} \left( H_{,\theta}(k-\half) + H_{,\theta}(k+\half) \right) = 0 \\
& \tfrac{1}{2} A_{\theta,\phi}(k-\half) (\theta_k - \theta_{k-1}) + \tfrac{1}{2} A_{\theta, \phi}(k+\half) (\theta_{k+1} - \theta_k) + \nonumber \\
& \tfrac{1}{2} A_{\phi,\phi}(k-\half) (\phi_k - \phi_{k-1}) + \tfrac{1}{2} A_{\phi, \phi} (k+\half) (\phi_{k+1} - \phi_k) + \nonumber \\
& A_\phi(k-\half) + u_{k-1/2} - A_\phi(k+\half) - u_{k+1/2} - \tfrac{h}{2} \left( H_{,\phi}(k-\half) + H_{,\phi}(k+\half) \right) = 0 \\
& \phi_{k+1} - \phi_k - \tfrac{h}{2} H_{,u}(k+\half) = 0,
\end{align}
\end{subequations}
where $(k+1/2)$ refers to $(r_{k+1/2},(\theta_{k}+\theta_{k+1})/2,(\phi_{k}+\phi_{k+1})/2)$ and similarly for $(k-1/2)$.

For the trapezoidal DVI scheme, we use the discrete Lagrangian in Eq.~\eqref{centered_LD-prime}, namely
\begin{equation}
\begin{split}
&L_d(r_{k+1/2}, \theta_{k}, \theta_{k+1}, \phi_k, \phi_{k+1}, u_{k+1/2}) = \nonumber \\
& \quad \half \left( A_\theta(r_{k+1/2}, \theta_k, \phi_k) + A_\theta(r_{k+1/2}, \theta_{k+1}, \phi_{k+1}) \right) (\theta_{k+1} - \theta_k) +\nonumber \\
& \quad \half \left( A_\phi(r_{k+1/2}, \theta_k , \phi_k ) + A_\phi(r_{k+1/2}, \theta_{k+1} , \phi_{k+1} ) 
 + 2 u_{k+1/2} \right) (\phi_{k+1} - \phi_k) - \nonumber \\
& \quad \tfrac{h}{2} \left( H(r_{k+1/2}, \theta_k , \phi_k, u_{k+1/2} ) + H(r_{k+1/2}, \theta_{k+1} , \phi_{k+1}, u_{k+1/2})\right].
 \end{split}
\end{equation}
The guiding center TDVI update equations follow from straightforward variations of this discrete Lagrangian (or from Eq.~\eqref{tdvi-noncanonical}) and shall be omitted here for brevity.

\begin{figure}
    \centering
    \includegraphics[width=0.7\columnwidth]{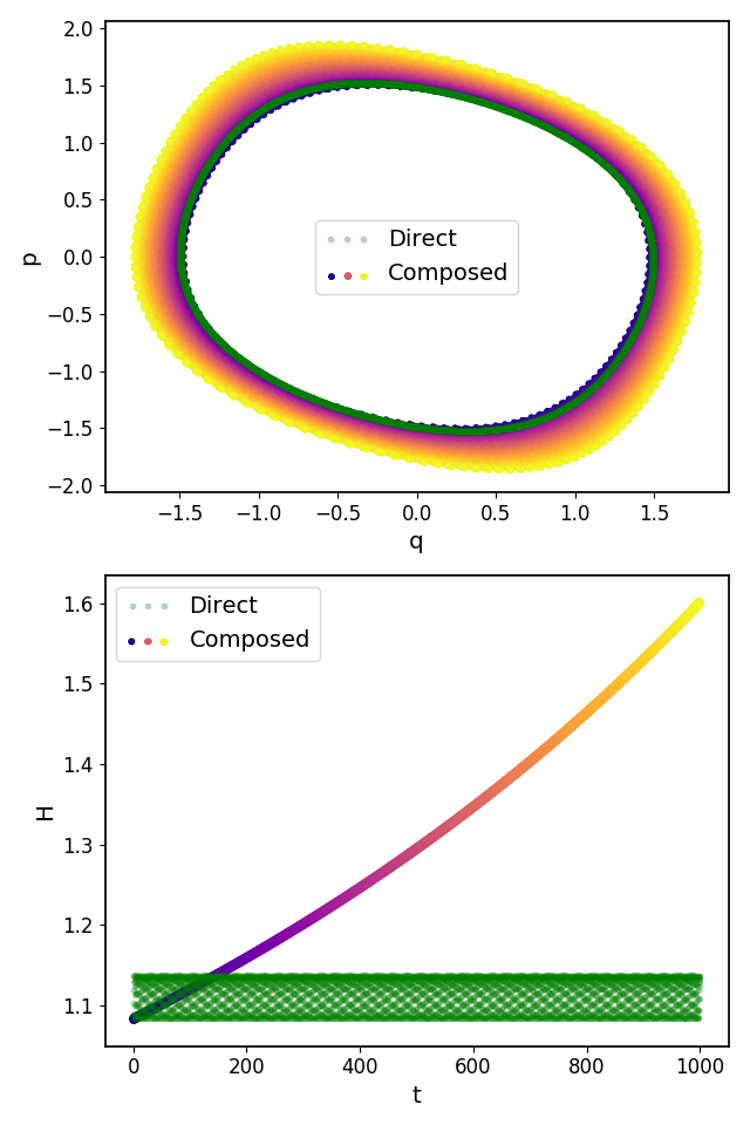}
    \caption{Phase portrait (top) and energy vs time (bottom) for two different methods of integrating the non-reversible Hamiltonian $H=(p^{2}+q^{2})/2+\alpha q p^{3}/3$.
    The `direct' scheme of Eq.\,\eqref{eq:Mod-SE-scheme}  and shown in green, is first-order accurate and preserves a discrete symplectic structure that depends on the numerical time step size $h$, with similar results for the adjoint scheme in Eq.\,\eqref{eq:Adj-Mod-SE-scheme}. Composing the direct method with its adjoint achieves second-order accuracy by virtue of centering in time, giving smaller oscillations in $H$. However, the numerical results indicate that the composed scheme leads to growth in $H$ and therefore does not preserve a two-form.}
    \label{fig:my_label-AA_adjoint}
\end{figure}

\section{Numerical tests\label{sec:examples}}
In this section we present numerical tests showing that the MDVI and TDVI schemes provide all the advantages of variational schemes and are second order accurate.

\subsection{Failure of composing with the adjoint}\label{Numerical-no-go-composing}

In Sec. \ref{Sec.Composing} we noted an example of a scheme with a preserved two-form with $\omega(h)\neq \omega(-h)$, suggesting complications if the scheme and its adjoint are composed in an attempt to obtain second order accuracy. In this section we give a concrete example for which the composed scheme appears not to have a preserved two-form. We consider the non-reversible Hamiltonian
 \begin{align}
 H=\frac{p^{2}+q^{2}}{2}+\frac{\alpha q p^{3}}{3}.\label{Non-rev-Hamiltonian}
 \end{align}
  Numerical tests show that both of the schemes of Eqs.~(\ref{eq:Mod-SE-scheme})
 and (\ref{eq:Adj-Mod-SE-scheme}) have good long-time behavior for this Hamiltonian. The orbits
 of the two first order schemes have bounded behavior but, with
 first order accuracy, have noticeable oscillations in the value
 of $H$ between bounds. The composed scheme, which has smaller oscillations in $H$, otherwise performs poorly: the points that should stay near 
 the surface constant $H$ surfaces spiral out, as shown in Fig.\,\ref{fig:my_label-AA_adjoint}. The behavior of $H$ in time shows exponential increase, with $\gamma=O(h^2)$. Appendix B shows an example for which composition does appear to give useful results, but this particular example is a reversible system, and this reversibility by itself appears to be responsible for the positive results.

\subsection{Magnetic field line\label{MFL-numerical}}

To test the proposed algorithms in a magnetic configuration representative of those of interest to the magnetic fusion community, we use the simple analytic expression for an axisymmetric, toroidal magnetic field presented in Ref.\,\onlinecite{Qin_2009}:
\begin{equation}
    \vec{A}(r, \theta, \phi) = \frac{B_0 R_0}{\cos^2 \theta} \left(r \cos \theta -R_0 \log \left( 1 + \frac{r \cos \theta}{R_0}\right) \right)\nabla \theta 
\end{equation}
\begin{equation}
- \frac{B_0 r^2}{2 q_0} \nabla \phi,   
    \label{eq:vecpotential_axisym}
\end{equation}
where $B_0$ is a magnetic field amplitude, $R_0$ is the major radius and $q_0=\sqrt{2}$ is the on-axis safety factor. The variables $(x^1,x^2,x^3)$ of Eq.~\eqref{MFL-z-action} are replaced by simple toroidal coordinates $(r,\theta,\phi)$, and $\phi$ takes the place of the time variable, as discussed in Sec.~\ref{MFL}.

\begin{figure}
    \centering
    \includegraphics[width=0.8\columnwidth]{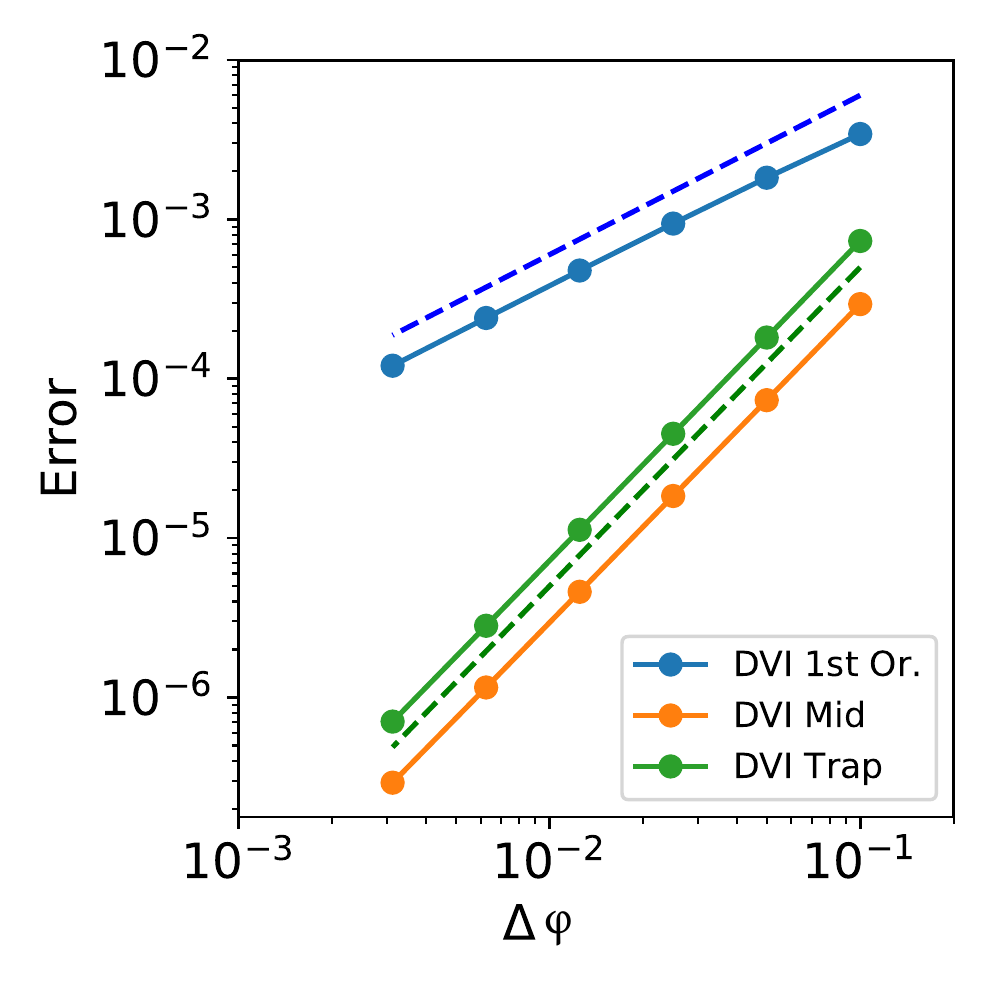}
    \caption{Comparison of the order of accuracy of the MDVI scheme, the TDVI scheme, and the collocated first-order accurate scheme of Ref.~\onlinecite{Ellison_2018} for the magnetic field line problem, over a range of $2^5$ in $\Delta \phi$. The dotted curves show curves Error$\sim h$ (blue) and Error$\sim h^2$ (green) for comparison.}
    \label{fig:my_labelBField_conv}
\end{figure}

First, we demonstrate numerically that the MDVI and TDVI achieve the anticipated second-order accuracy. Next, we demonstrate that the proposed algorithms exhibit the expected qualitative behavior of symplectic integrators. To the axisymmetric magnetic field of Eq.~\eqref{eq:vecpotential_axisym} we add a perturbation of the form:
\begin{equation}
    \vec{A}(r, \theta, \phi) = \vec{A}_0(r, \theta, \phi) - \frac{B_0 r^2}{2 q_0}\sum_i \delta_i \sin(m_i \theta - n_i \phi)\nabla \phi.
\end{equation}
where $\vec{A_0}$ is given by Eq.~\eqref{eq:vecpotential_axisym}. We choose two perturbative harmonics, $m_0=3, n_0=2$ and $m_1=7, n_1=5$ with amplitudes $\delta_0 = \delta_1 = 10^{-4}$. These perturbations lead to magnetic islands at the resonant magnetic surfaces, and small stochastic field line regions in the $(m_1,n_1)=(3,2)$ and $(m_2,n_2)=(7,5)$ resonant regions. 

We test the order of accuracy by varying the step $\Delta \phi$ in factors of two across a range of $2^5$ and comparing with a fourth-order Runge-Kutta scheme with an extremely small value of $\Delta \phi$. See Fig.\,\ref{fig:my_labelBField_conv}. We compare the MDVI and TDVI schemes with each other and with the first order variant described in \onlinecite{Ellison_2018} in which the three variables $(r,\theta,\phi)$ are collocated, i.e.~not staggered. We integrate over a large number of toroidal transits for this comparison. Second order accuracy is confirmed for the MDVI and TDVI schemes, as is first order accuracy of the non-staggered scheme. The MDVI and TDVI schemes exhibit relatively similar accuracy, with the MDVI being more accurate for this particular example. 

The Poincar\'e surface of section at $\phi=0$ is shown in Fig.\,\ref{fig:my_label_B_Poincare}. The first two panels are for the second order Runge-Kutta (RK2) scheme, with $\Delta \phi=0.1$ and $\Delta \phi=5\times 10^{-4}$, respectively. last two panels are for the MDVI scheme for the same two values of $\Delta \phi$. The RK2 results in Fig.\,\ref{fig:my_label_B_Poincare}(a) show blurriness of the KAM surfaces, falsely indicating a higher degree of magnetic stochasticity. The results in (b) are greatly improved. The two MDVI cases in (c) and (d), for very different steps $\Delta \phi$, look almost identical, showing very good preservation of KAM tori. The TDVI scheme leads to results essentially indistinguishable from those of the MDVI scheme for comparable time steps.

\begin{figure*}
    \centering
    \includegraphics[width=0.8\textwidth]{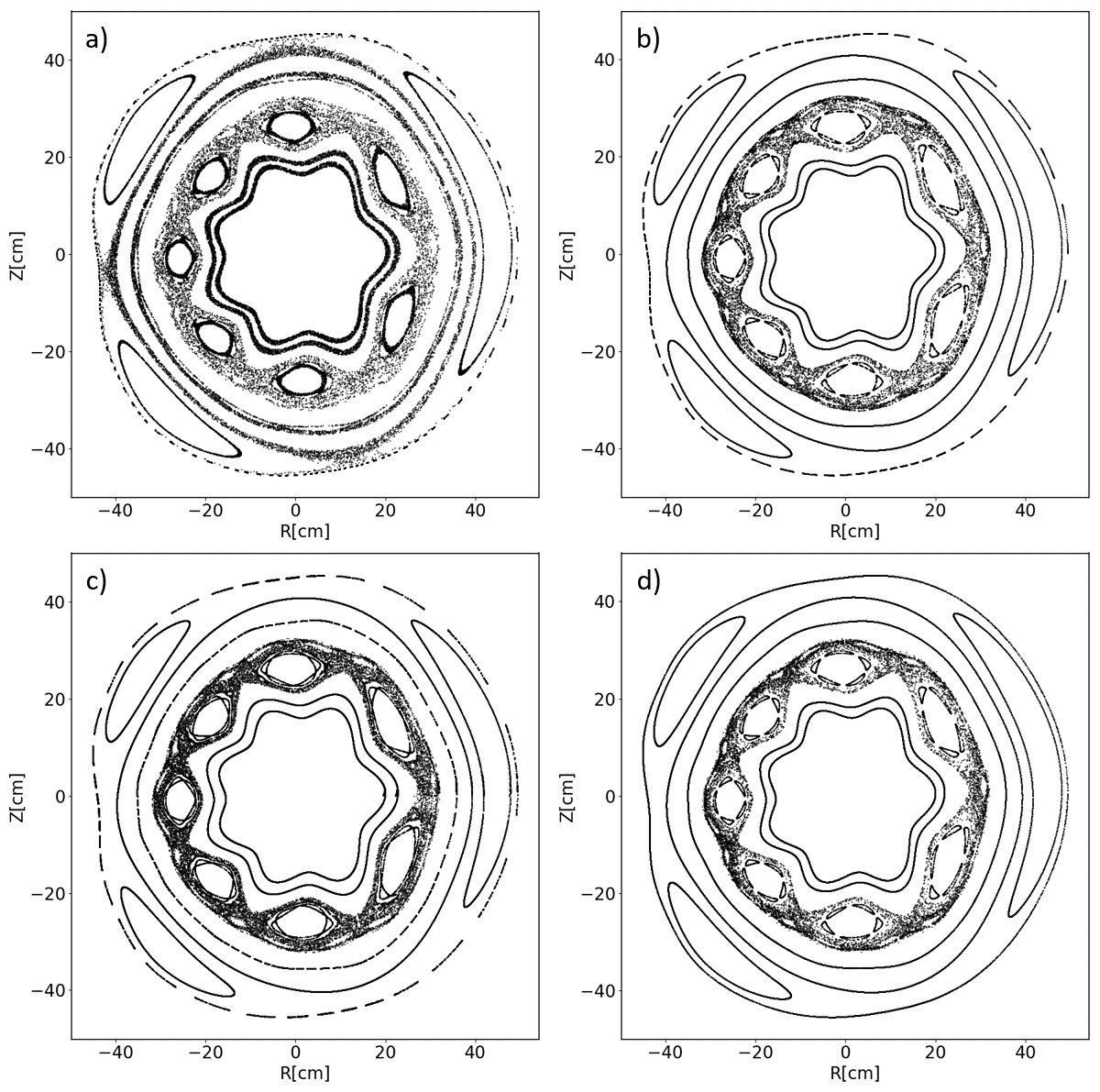}
    \caption{Poincar\'e Section $\phi=0$ for the magnetic field line problem. For comparison we show results from a) the second order accurate Runge-Kutta (RK2) scheme with $\Delta \phi = 0.1$, b) the RK2 scheme with $\Delta \phi = 0.005$, c) the MDVI scheme with $\Delta \phi = 0.1$, and  d) the MDVI scheme with $\Delta \phi = 0.005$. In all cases we have $\phi_\text{final} = 3\times 10^5$. The $(m,n)=(3,2)$ and $(7,5)$ island chains are evident, and in the RK2 scheme in (a) the KAM surfaces are blurred.}
    \label{fig:my_label_B_Poincare}
\end{figure*}

\begin{figure}
    \centering
    \includegraphics[width=0.9\columnwidth]{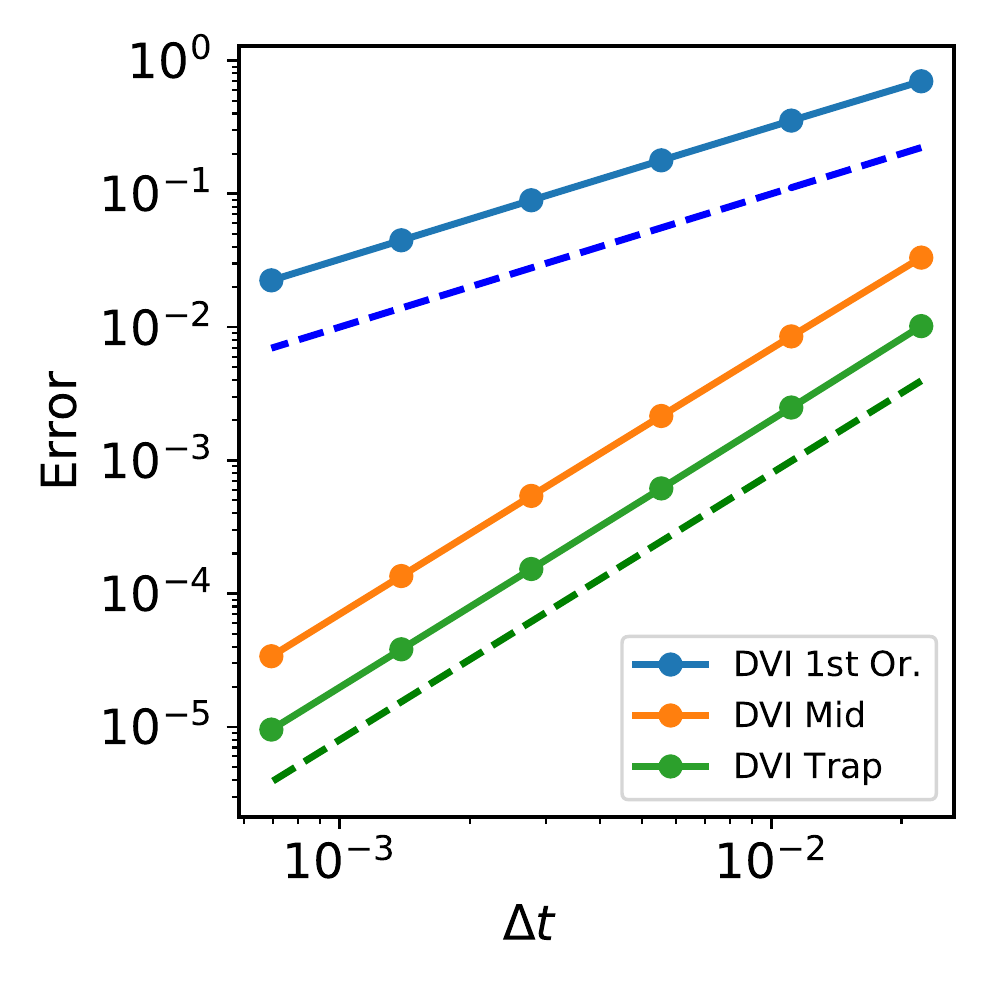}
    \caption{ Error in the second order schemes MDVI and TDVI and the first order collocated scheme of Ref.~\onlinecite{Ellison_2018}, compared with `exact' results, i.e.~results obtained with a fourth-order Runge-Kutta method with extremely small time step $h=\Delta t$. Results are shown over almost two orders of magnitude in $\Delta t =h$. First order and second order reference lines, in blue dashed and green dashed, respectively, are shown for comparison. Second order accuracy for the MDVI and TDVI schemes is confirmed.}
    \label{fig:my_label_GC_convergence}
\end{figure}


\subsection{Guiding center \label{GC-label}}
In this section we show numerical results for the guiding center example, with time independent potentials and scalar potential $\varphi=0$. That is, physically, the electric field is zero. See Fig.\,\ref{fig:my_label_GC_convergence}, confirming second order accuracy in $h=\Delta t$ for the MDVI and TDVI schemes.  Interestingly, the error in the MDVI scheme about a factor of $3$ larger, in contrast with the results shown in Fig.\,\ref{fig:my_labelBField_conv}, where the MDVI and TDVI results are reversed. As for the magnetic field line results, this difference is due to the very long run times.   


\section{Non-uniform time stepping\label{sec:Non-unifTS}}

We first review the variational form of the extended phase space method
for a Hamiltonian system in canonical variables. This method allows
us to prescribe variable time steps. We show an example, of discretizing this
action by the symplectic Euler scheme, and show that it is a single-step
method, i.e.~a degenerate variational integrator. We proceed to formulate
an extended phase space method for systems having the form of noncanonical
variables prescribed in Eq.~\eqref{usual_suspect}, giving examples
of magnetic field line integration and the guiding center equations.
For the most straightforward first order accurate scheme with non-uniform time steps, we exhibit an analog of the modified symplectic Euler schemes of Eqs.~\eqref{eq:Mod-SE-scheme} and \eqref{eq:Adj-Mod-SE-scheme}. 
We show that this method is a DVI, connecting only two time levels,
according to the discrete Hessian method as well as by a direct substitution.
We also discuss extending these non-uniform time step methods to second
order accuracy, as well as using adaptive time step control based
on an error estimator.

\subsection{Canonical systems\label{subsec:Nonuniform-time-stepping-canonical}}

We first review the extended phase space method in canonical variables,
in one degree of freedom for transparency. The time-stepping is defined
by a time step density $\rho(q,p)$, with $\rho(q,p)dt=d\zeta$. We
prescribe uniform steps in the time-like variable $\zeta$, with $\rho\Delta t\approx\Delta\zeta=h=\text{const}$.
We extend the action $S=\int(p\,dq/dt-H)dt$ for the Hamiltonian $H=H(q,p)$
to 
\begin{equation}
S_{e}=\int\left[p\frac{dq}{d\zeta}-\frac{H(q,p)}{\rho(q,p)}+\pi\left(\frac{dw}{d\zeta}-\frac{1}{\rho(q,p)}\right)\right]d\zeta.\label{eq:Ad-act-Lagrange-multiplier}
\end{equation}
Here we have written time as the dependent variable $w$, and the
Lagrange multiplier $\pi$ enforces the time step condition as a constraint.
We rewrite this as
\begin{equation}
S_{e}=\int\left[p\frac{dq}{d\zeta}+\pi\frac{dw}{d\zeta}-\frac{H(q,p)+\pi}{\rho(q,p)}\right]d\zeta,\label{eq:Ad-act-canonical-form}
\end{equation}
leading to the extended phase space Hamiltonian $K(q,p,\pi)=\left(H(q,p)+\pi\right)/\rho(q,p)$,
with an added canonically conjugate pair, $(q,p)\to(q,p,w,\pi)$.
It is clear that, since $K$ does not depend on $\zeta$ explicitly,
$K$ is exactly conserved. That is, if we set $\pi=-H$ initially,
then $H+\pi$ remains exactly zero. Thus, the equations for $(q,p,w,\pi)$,
using $H+\pi=0$, are 
\[
\frac{dq}{d\zeta}=\frac{H_{2}(q,p,w)}{\rho(q,p)},\,\,\frac{dp}{d\zeta}=-\frac{H_{1}(q,p,w)}{\rho(q,p)},
\]
\[
\frac{dw}{d\zeta}=\frac{1}{\rho(q,p)},\,\,\frac{d\pi}{d\zeta}=0.
\]
Note that the imposed time step requirement is satisfied. The extension
to a non-autonomous system, with $H(q,p,t)$ is straightforward.

For discrete integration, it is important to keep the terms (not shown)
derived from $H+\pi$ in Eq.~\eqref{eq:Ad-act-canonical-form}: this quantity is not exactly zero for the
discrete equations and if this quantity is set exactly equal to zero,
the symplectic nature is lost\cite{Richardson_2011}.

To illustrate, we discretize $S_{e}$ as in the symplectic Euler scheme,
here for two degrees of freedom, for uniform stepping in $\zeta$, $\Delta \zeta=h$. We have
\[
S_{e1}=\sum_{k=0}^{N-1}\left[p_{k}(q_{k+1}-q_{k})+\pi_{k}(w_{k+1}-w_{k})\right]
\]
\begin{equation}
-\sum_{k=0}^{N-1} \left[h\frac{H(q_{k+1},p_{k},w_{k+1})+\pi_{k}}{\rho(q_{k+1},p_{k})}\right].\label{eq:SE-Action-Extended}
\end{equation}


The symplectic Euler scheme derived from Eq.~(\ref{eq:SE-Action-Extended})
preserves the canonical two degree of freedom two-form
\begin{equation}
\omega=dq\wedge dp+dw\wedge d\pi,\label{eq:EPS-SE-two-form}
\end{equation}
by inspection or by taking the endpoint values of $dS_{e1}$.
This scheme is a single-step scheme, shown either by inspection or
by computing the discrete Hessian.

We now consider the special class of systems in noncanonical variables
with action of the form $\int(f_{i}\dot{x}_{i}-H(\boldsymbol{x},\boldsymbol{y}))dt$,
as in Eq.~\eqref{usual_suspect}. With time step condition $\rho(\boldsymbol{x},\boldsymbol{y})dt=d\zeta$
and $t\to w$ we can write the analog to Eqs.~(\ref{eq:Ad-act-Lagrange-multiplier},\ref{eq:Ad-act-canonical-form}),
\[
S=\int\left[f_{i}(\boldsymbol{x},\boldsymbol{y})\frac{dx^{i}}{d\zeta}-\frac{H(\boldsymbol{x},\boldsymbol{y})}{\rho(\boldsymbol{x},\boldsymbol{y})}+\pi\left(\frac{dw}{d\zeta}-\frac{1}{\rho(\boldsymbol{x},\boldsymbol{y})}\right)\right]d\zeta
\]
\begin{equation}
=\int\left[f_{i}\frac{dx^{i}}{d\zeta}+\pi\frac{dw}{d\zeta}-\frac{H(\boldsymbol{x},\boldsymbol{y},w)+\pi}{\rho(\boldsymbol{x},\boldsymbol{y})}\right]d\zeta.\label{eq:General_fdx_nonunif}
\end{equation}
Again, it is evident in the first form that $\pi$ is a Lagrange multiplier
enforcing the time step condition. We can apply uniform stepping in
the new time-like variable $\zeta$, with $\Delta\zeta=h=$const., to
give the required non-uniform time stepping. In Eq.~(\ref{eq:General_fdx_nonunif})
a term involving the canonical pair $(w,\pi)$, namely $\pi dw/d\zeta$,
is added to $f_{i}dx^{i}/d\zeta$. Furthermore, the new Hamiltonian
is $K=(H+\pi)/\rho$, and the resulting action is of the same restricted
noncanonical class of Eq.~\eqref{usual_suspect}. Therefore any discretization that can be applied
to the action with uniform time step can be applied to this noncanonical
extended phase space version.

\subsection{Magnetic field line integration\label{sec:MagFieldLine}}

For the magnetic field line integration problem, we use a gauge with
$A_{1}=0$, as in Sec.~\ref{MFL}. For nonuniform time
stepping, we first go back to considering $x^3$ to be a coordinate
and put the action in the form in Eq.~(\ref{eq:General_fdx_nonunif}),
leading to
\begin{equation}
\Phi_{e}=\int\left[A_{2}(\boldsymbol{x})\frac{dx^2}{d\zeta}+\frac{A_{3}(\boldsymbol{x})}{\rho(\boldsymbol{x})}+\pi\left(\frac{dx^3}{d\zeta}-\frac{1}{\rho(\boldsymbol{x})}\right)\right]d\zeta,\label{eq:MFL-action-adaptive}
\end{equation}
where $\pi$ is a Lagrange multiplier enforcing the time step restriction and $\boldsymbol{x}=(x^1,x^2,x^3)$.
This can again be put into the form
\begin{equation}
\Phi_{e}=\int\left[A_{2}(\boldsymbol{x})\frac{dx^2}{d\zeta}+\pi\frac{dx^3}{d\zeta}-\frac{-A_{3}(\boldsymbol{x})+\pi}{\rho(\boldsymbol{x})}\right]d\zeta.\label{eq:MFL-action-adaptive-2}
\end{equation}
This action if of the form in Eq.~(\ref{eq:General_fdx_nonunif}),
with two degree of freedom Hamiltonian equal to $(-A_{3}+\pi)/\rho$.

We discretize this in a manner similar to the modified (adjoint) symplectic
Euler scheme described in Eq.~\eqref{eq:Mod-SE-scheme} by forming $\Phi_{e1}=\sum_{k}hL_{d}(\boldsymbol{x}_{k},\pi_{k},\boldsymbol{x}_{k+1},\pi_{k+1})$,
or
\[
\Phi_{e1}=\sum_{k}\left[A_{2}(\boldsymbol{x}_{k+1})(x^2_{k+1}-x^2_{k})+\pi_{k+1}(x^3_{k+1}-x^3_{k})\right]
\]
\begin{equation}
+\sum_k \left[h\frac{A_{3}(\boldsymbol{x}_{k+1})-\pi_{k+1}}{\rho(\boldsymbol{x}_{k+1})}\right].\label{eq:MFL-discrete-action-rho}
\end{equation}

Compared with Eq.~\eqref{eq:Mod-SE-scheme}, $p_{k+1}\to A_{2}(\boldsymbol{x}_{k+1})$
and $H(q_{k+1},p_{k+1})\to A_{3}(\boldsymbol{x}_{k+1})$ with more
direct substitutions for $x^3_{k}$ and $\pi_{k}$. 

The DEL equations from $d\Phi_{e1}=0$ lead to the fourth order system
\[
A_{2,1}(\boldsymbol{x}_{k})(x^2_{k}-x^2_{k-1})+h\frac{A_{3,1}(\boldsymbol{x}_{k})}{\rho(\boldsymbol{x}_{k})}
\]
\begin{equation}
+h\left[A_{3}(\boldsymbol{x}_{k})-\pi_{k}\right]\frac{\partial}{\partial x^1_{k}}\frac{1}{\rho(\boldsymbol{x}_{k})}=0,\label{eq:MFLxxx-adaptive-penult}
\end{equation}
\[
A_{2,2}(\boldsymbol{x}_{k})(x^2_{k}-x^2_{k-1})+A_{2}(\boldsymbol{x}_{k})-A_{2}(\boldsymbol{x}_{k+1})
\]
\begin{equation}
+h\frac{A_{3,2}(\boldsymbol{x}_{k})}{\rho(\boldsymbol{x}_{k})}+h\left[A_{3}(\boldsymbol{x}_{k})-\pi_{k}\right]\frac{\partial}{\partial x^2_{k}}\frac{1}{\rho(\boldsymbol{x}_{k})}=0,\label{eq:MFLyyy-adaptive-penult}
\end{equation}
\[
A_{2,3}(\boldsymbol{x}_{k})(x^3_{k}-x^3_{k-1})+\pi_{k}-\pi_{k+1}+h\frac{A_{3,3}(\boldsymbol{x}_{k})}{\rho(\boldsymbol{x}_{k})}
\]
\begin{equation}
+h\left[A_{3}(\boldsymbol{x}_{k})-\pi_{k}\right]\frac{\partial}{\partial x^3_{k}}\frac{1}{\rho(\boldsymbol{x}_{k})}=0,\label{eq:MFLzzz-adaptive-penult}
\end{equation}
\begin{equation}
x^3_{k}-x^3_{k-1}-\frac{h}{\rho(\boldsymbol{x}_{k})}=0.\label{eq:MFLppp-adaptive-penult}
\end{equation}
For the uniform time stepping case of Sec.~\ref{MFL} the assumption $\rho=\rho(\boldsymbol{x})$ means that the time step
depends on $x^3$, which
takes the place of time.
This suggests the possibility that complications such as parametric instabilities related to having a time step density explicitly dependent on time might arise\cite{Richardson_2011}. In the case treated in this subsection, on the other hand, $\zeta$ rather than $z$ is the independent (time-like) variable.

As occurred in Ref.~\onlinecite{Ellison_2018}, equations (\ref{eq:MFLyyy-adaptive-penult},\ref{eq:MFLzzz-adaptive-penult})
appear to involve indices $k-1,\,\,k,\,\,k+1$ and therefore appear
to be two-step equations, suggesting that Eqs.~(\ref{eq:MFLxxx-adaptive-penult}-\ref{eq:MFLppp-adaptive-penult})
are difference equations of order higher than $4$. However, the discrete
Hessian can be shown to have rank $4$, consistent with a first order system in $(\boldsymbol{x},\pi)$; indeed similar substitutions
as those of Ref.~\onlinecite{Ellison_2018} lead to a fourth order system,
i.e.~a single-step method. That is, writing Eqs.~(\ref{eq:MFLyyy-adaptive-penult},\ref{eq:MFLxxx-adaptive-penult},\ref{eq:MFLzzz-adaptive-penult})
in the compact form
\begin{equation}
A_{2,1}(\boldsymbol{x}_{k})(x^2_{k}-x^2_{k-1})+hP(\boldsymbol{x}_{k},\pi_{k})=0,\label{eq:MFLxxx/simp}
\end{equation}
\begin{equation}
A_{2,2}(\boldsymbol{x}_{k})(x^2_{k}-x^2_{k-1})+A_{2}(\boldsymbol{x}_{k})-A_{2}(\boldsymbol{x}_{k+1})+hQ(\boldsymbol{x}_{k},\pi_{k})=0,\label{eq:MFLyyy/simp}
\end{equation}
\begin{equation}
A_{2,3}(\boldsymbol{x}_{k})(x^2_{k}-x^2_{k-1})+\pi_{k}-\pi_{k+1}+hR(\boldsymbol{x}_{k},\pi_{k})=0,\label{eq:MFLzzz/simp}
\end{equation}
we find $x^2_{k}-x^2_{k-1}$ from Eq.~(\ref{eq:MFLxxx/simp}) and substitute
into Eqs.~(\ref{eq:MFLyyy/simp},\ref{eq:MFLzzz/simp}). When the
indices are incremented $k\rightarrow k+1$ in Eqs.~(\ref{eq:MFLxxx/simp},\ref{eq:MFLppp-adaptive-penult})
we find
\[
A_{2,1}(\boldsymbol{x}_{k+1})(x^2_{k+1}-x^2_{k})+hP(\boldsymbol{x}_{k+1},\pi_{k+1})=0,
\]
\[
-hA_{2,2}(\boldsymbol{x}_{k})\frac{P(\boldsymbol{x}_{k},\pi_{k})}{A_{2,1}(\boldsymbol{x}_{k})}+A_{2}(\boldsymbol{x}_{k})-A_{2}(\boldsymbol{x}_{k+1})+hQ(\boldsymbol{x}_{k},\pi_{k})=0,
\]
\[
-hA_{2,3}(\boldsymbol{x}_{k})\frac{P(\boldsymbol{x}_{k},\pi_{k})}{A_{2,1}(\boldsymbol{x}_{k})}+\pi_{k}-\pi_{k+1}+hR(\boldsymbol{x}_{k},\pi_{k})=0,
\]
\begin{equation}
x^3_{k+1}-x^3_{k}-\frac{h}{\rho(\boldsymbol{x}_{k+1})}=0.\label{eq:MFLppp-adaptive-penult-1}
\end{equation}
Because only indices $k$ and $k+1$ are involved, the single-step property
predicted by the rank of the discrete Hessian is evident. 
Note the
solvability condition $A_{2,1}=B^{3}\neq0$, necessary for $z$ to
parameterize the length along the field line. 


\subsection{Guiding center equations\label{subsec:Guiding-center-equations}}



Guided by the results for the magnetic field line equations, it was
shown in Ref.~\onlinecite{Ellison_2018} that if the term $A_{1}^{\dagger}\dot{x}^1$
is zero, substitutions can be made to lead to a single-step scheme.
Because the requirement $A_{1}^{\dagger}(\boldsymbol{x},u)=A_{1}(\boldsymbol{x})+ub_{1}(\boldsymbol{x})=0$
must hold for arbitrary values of $u$, it requires both the gauge
condition $A_{1}=0$ and the physical condition $b_{1}=0$. In Ref.~\onlinecite{Ellison_2018}, numerical tests were performed for axisymmetric fields
and for coordinates such that the covariant component $b_{1}$ is
zero. In Ref.~\onlinecite{Burby_Ellison_2017}, it was shown that it is possible
to find coordinates such that this condition is satisfied for arbitrary
magnetic fields, provided one component, e.g.~the toroidal component
does not change sign.

For applying nonuniform time steps to the guiding center equations,
we again introduce a new time-like variable $\zeta$ such that $\rho(\boldsymbol{x},u)dt=d\zeta$
for the time step density $\rho(\boldsymbol{x},u)$. Then we make
the substitution, with $w=t$ and
\[
S_{gc}=\int\left[A_{2}^{\dagger}(\boldsymbol{x})\frac{dx^2}{dt}+A_{3}^{\dagger}(\boldsymbol{x})\frac{dx^3}{dt}-H_{gc}(\boldsymbol{x},u)\right]dt
\]
going to
\begin{equation}
S_{gc}=\int\left[A_{2}^{\dagger}(\boldsymbol{x})\frac{dx^2}{d\zeta}+A_{3}^{\dagger}(\boldsymbol{x})\frac{dx^3}{d\zeta}-\frac{H_{gc}(\boldsymbol{x},u)}{\rho(\boldsymbol{x},u)}\right] d\zeta\label{eq:Action-GC-nonunif-1}
\end{equation}
\[
+\int \left[\frac{dw}{d\zeta}-\frac{1}{\rho(\boldsymbol{x},u)}\right]d\zeta.
\]
Again this can be put in the form
\begin{equation}
S_{gc}=\int\left[A_{2}^{\dagger}(\boldsymbol{x})\frac{dx^2}{d\zeta}+A_{3}^{\dagger}(\boldsymbol{x})\frac{dx^3}{d\zeta}
+\pi\frac{dw}{d\zeta}-K(\boldsymbol{x},u,\pi)\right]d\zeta,\label{eq:Action-GC-nonunif-2}
\end{equation}
where $K(\boldsymbol{x},u,p_{0})=\left(H_{gc}(\boldsymbol{x},u)+\pi\right)/\rho(\boldsymbol{x},u)$.
This is an action in the form of Eq.~(\ref{eq:General_fdx_nonunif}) 
on the extended phase space $(\boldsymbol{x},u)\rightarrow(\boldsymbol{x},u,w,\pi)$.
We discretize this action in a manner similar to that in Eq.~(\ref{eq:MFL-discrete-action-rho}),
namely 
\begin{equation}
S_{gc1}=\sum_{k}hL_{d}(\boldsymbol{x}_{k},u_{k},w_{k},\pi_{k},\boldsymbol{x}_{k+1},u_{k+1},w_{k+1},\pi{}_{k+1}),\label{eq:S-gc-1}
\end{equation}
 with
\begin{equation}
hL_{d}=A_{2}^{\dagger}(\boldsymbol{x}_{k+1},u_{k+1})(x^2_{k+1}-x^2_{k})+A_{z}^{\dagger}(\boldsymbol{x}_{k+1},u_{k+1})(x^3_{k+1}-x^3_{k})\label{eq:GC-discrete-action-1}
\end{equation}
\[
+\pi_{k+1}(w_{k+1}-w_{k})-hK(\boldsymbol{x}_{k+1},u_{k+1},\pi_{k+1})
\]
where
\begin{equation}
K(\boldsymbol{x}_{k+1},u_{k+1},\pi_{k+1})=\frac{H_{gc}(\boldsymbol{x}_{k+1},u_{k+1})+\pi_{k+1}}{\rho(\boldsymbol{x}_{k+1},u_{k+1})}.\label{eq:GC-EPS-Hamiltonian}
\end{equation}
The DEL equations are
\begin{equation}
A_{2,1}^{\dagger}(\boldsymbol{x}_{k},u_{k})(x^2_{k}-x^2_{k-1})+A_{3,1}^{\dagger}(\boldsymbol{x}_{k},u_{k})(x^3_{k}-x^3_{k-1})
\end{equation}
\begin{equation}
    -hK_{1}(\boldsymbol{x}_{k},u_{k},\pi_{k})=0,\label{eq:GC-DELx}
\end{equation}
\[
A_{2,2}^{\dagger}(\boldsymbol{x}_{k},u_{k})(x^2_{k}-x^2_{k-1})+A_{3,2}^{\dagger}(\boldsymbol{x}_{k},u_{k})(x^3_{k}-x^3_{k-1})
\]
\begin{equation}
+A_{2}^{\dagger}(\boldsymbol{x}_{k},u_{k})-A_{2}^{\dagger}(\boldsymbol{x}_{k+1},u_{k+1})-hK_{2}(\boldsymbol{x}_{k},u_{k},\pi_{k})=0,\label{eq:GC-DELy}
\end{equation}
\[
A_{2,3}^{\dagger}(\boldsymbol{x}_{k},u_{k})(x^2_{k}-x^2_{k-1})+A_{3,3}^{\dagger}(\boldsymbol{x}_{k},u_{k})(x^3_{k}-x^3_{k-1})
\]
\begin{equation}
+A_{3}^{\dagger}(\boldsymbol{x}_{k},u_{k})-A_{3}^{\dagger}(\boldsymbol{x}_{k+1},u_{k+1})-hK_{3}(\boldsymbol{x}_{k},u_{k},\pi_{k})=0,\label{eq:GC-DELz}
\end{equation}
\begin{equation}
b_{2,k}(x^2_{k}-x^2_{k-1})+b_{3,k}(x^3_{k}-x^3_{k-1})-hK_{u}(\boldsymbol{x}_{k},u_{k},\pi_{k})=0,\label{eq:GC-DELu}
\end{equation}
\begin{equation}
\pi_{k}-\pi_{k+1}=0,\label{eq:GC-DELx0}
\end{equation}
\[
w_{k}-w_{k-1}=h\frac{\partial}{\partial p_{w,k}}K(\boldsymbol{x}_{k},u_{k},\pi_{k})
\]
\begin{equation}
=\frac{h}{\rho(\boldsymbol{x}_{k},u_{k})}.\label{eq:GC-DEL-p0}
\end{equation}
The assumed time-independence of the fields leads to the simple form
in Eq.~(\ref{eq:GC-DELx0}). Similar to the uniform time step case
of Ref.~\onlinecite{Ellison_2018}, the discrete Hessian has rank six, consistent with a first order system in $(x,y,z,u,w,\pi)$. We start by taking Eqs.~(\ref{eq:GC-DELx}) and (\ref{eq:GC-DELu}),
written as
\[
\left[\begin{array}{cc}
A_{2,1}^{\dagger}(\boldsymbol{x}_{k},u_{k}) & A_{3,1}^{\dagger}(\boldsymbol{x}_{k},u_{k})\\
b_{2,k} & b_{3,k}
\end{array}\right]\left[\begin{array}{c}
x^2_{k}-x^2_{k-1}\\
x^3_{k}-x^3_{k-1}
\end{array}\right]
\]
\[=\left[\begin{array}{c}
hK_{1}(\boldsymbol{x}_{k},u_{k},w_{k},\pi_{k})\\
hK_{u}(\boldsymbol{x}_{k},u_{k},w_{k},\pi_{k})
\end{array}\right].
\]
Solving for $x^2_{k}-x^2_{k-1}$ and $x^3_{k}-x^3_{k-1}$, which are written
in terms of quantities with index $k$, we substitute these into Eqs.~(\ref{eq:GC-DELy})
and (\ref{eq:GC-DELz}), increment $k\rightarrow k+1$ in Eqs.~(\ref{eq:GC-DELy}),
(\ref{eq:GC-DELz}) and (\ref{eq:GC-DEL-p0}). The resulting equations
involve time steps labeled with only $k$ and $k+1$. That is, consistent
with the discrete Hessian condition, the scheme is a single-step scheme,
a DVI, and parasitic modes cannot occur. 


\subsection{Extensions for higher accuracy\label{sec:HigherAccuracy}}

From the formulation in the last two sections, it is clear from Eqs.~(\ref{eq:MFL-action-adaptive-2})
and (\ref{eq:Action-GC-nonunif-2}) that the modification to prescribe
nonuniform time stepping leads to an addition to the phase space Lagrangian
of a term $\pi dw/d\zeta$ or $\pi dz/d\zeta$ and a modification
to the Hamiltonian $H\to(H+\pi)/\rho$, and these terms can be discretized
in exactly the same manner as in the uniform time step case. This
means that the modifications in this section can be applied to any
discretization of the phase space Lagrangian that leads to a DVI.
Therefore, it should be straightforward to construct a nonuniform
time step scheme for either of the second order accurate DVI methods
of Sec.~\ref{sec:DVI2}.

It is also clear that such discretizations can be applied to any time
step density $\rho$, so that it should be straightforward to use
an optimum density $\rho$ based on an error estimator, to minimize
the integrated error over an orbit for the scheme at hand, as done
in Refs.~\onlinecite{Richardson_2011} and \onlinecite{Finn_2015}. Therefore, it is
possible to combine the formulations of this paper to give an adaptive
second order accurate variational integrator. We leave further details
to a future publication.

\section{Summary and discussion\label{sec:discussion}}

In previous work\cite{Ellison_2018,Ellison_thesis}, the concept of proper degeneracy for a discrete time-stepping scheme for a degenerate variational system was introduced. In these works, the focus was on systems governed by a phase space Lagrangian, which produces a system of first order differential equations, the Hamiltonian equations, in canonical or noncanonical variables. This concept relates to a discretization that preserves the first order nature of the Hamiltonian equations on phase space, i.e. is a single-step rather than a multistep scheme. Multistep schemes are to be avoided in variational systems because they can possess parasitic modes that can grow unphysically\cite{Ellison_2018}. For some examples, the single-step property can be determined by inspection simply. But it is in fact common to have a system that appears to have a multistep nature, but can be reduced to a form where the single-step property is evident. However, finding the right substitutions is not always so straightforward\cite{Ellison_2018}. In this reference, a method of addressing this single-step vs.~multistep issue in terms of the rank of the discrete Hessian was developed. In Refs.\,\onlinecite{Ellison_2018,Ellison_thesis}, schemes that preserve this single-step nature were called degenerate variational integrators or DVIs. 

The schemes developed in Ref.\,\onlinecite{Ellison_2018,Ellison_thesis} are all first order accurate. One aim of this paper is to develop second order accurate DVIs. A commonly used method of developing a second order accurate scheme from a first order variational scheme is a special case of a composition method\cite{Hairer_2006}. This involves composing the first order scheme $\Phi_{h}$ with its adjoint $\Phi_{h}^{\dagger}=\Phi_{-h}^{-1}$, and this method works well for discretizations that preserve the two-form $\omega_{0}$ of the original ODE system; this form is independent of $h$. However, for other schemes the discrete equations preserve a two-form that depends on the time step, $\omega=\omega(h)$. The adjoint of such a scheme preserves $\omega(-h)$ and it is not obvious whether the composed map preserves a two-form at all if $\omega(h)\neq\omega(-h)$. In this paper we consider an example of a simple autonomous Hamiltonian system in canonical variables, i.e. preserving the two-form $\omega_{0}=dq\wedge dp$ and a discretization of its phase space Lagrangian. This scheme preserves another form $\omega(h)=\omega_{0}+O(h)$, so that $\omega(h)\neq\omega(-h)$. Numerically, we find that, for some Hamiltonians, the orbits of the composed scheme spiral out with increasing $t$, the growth rate of the energy behaving like $\gamma=O(h^{2})$, showing that composing with the adjoint does not lead to a scheme with a preserved two-form in general, and therefore does not possess the advantageous properties of variational (symplectic) integration.

In the place of the composition method, we have constructed two centered schemes, involving a processing scheme to advance some of the variables to the half time step, and centering the other variables either in a midpoint or a trapezoidal manner. We call these schemes the midpoint DVI (MDVI) scheme and the trapezoidal DVI (TDVI) scheme. We have shown these schemes to be second order accurate by a backward error analysis and derived the properly degenerate property by computing the rank of the discrete Hessian (as well as by inspection.) We have also applied the midpoint and trapezoidal DVI schemes to two systems of importance to plasma physics, namely the magnetic field line system and the guiding center system. Both of these systems are in a restricted class of noncanonical variables. The numerical results show the anticipated positive properties, namely the benefits of degenerate variational integration, the lack of parasitic modes, and second order accuracy.

The second aim of this paper relates to using non-uniform time steps. This method has been developed for Hamiltonian systems in canonical variables\cite{Hairer_2006}. In this paper we show how to write a variational principle in extended phase space for systems with this class of noncanonical variables. Further, using an error estimator, it is possible to make the time step \emph{adaptive}, by minimizing the total error along an orbit, as in Ref.\,\onlinecite{Richardson_2011}.

We have first reviewed the extended phase space action
principle for one degree of freedom Hamiltonian systems in canonical variables with action
$S=\int\left(p\dot{q}-H\right)dt$, allowing variable time steps.
For canonical variables, this method involves a discretization of the
action with a constraint related to the variation of the time stepping,
producing a canonical symplectic integrator in the extended phase
space $(q,p)\to(q,p,w,\pi)$, where the extra canonical pair are time
and its canonical conjugate. The extension to noncanonical variables
applies to the restricted class of systems discussed earlier, with
variables $(\boldsymbol{x},\boldsymbol{y})$ and an action of the
form $\int(f_{i}\dot{x}_{i}-H(\boldsymbol{x},\boldsymbol{y}))dt$.
The two well-known examples of Hamiltonian systems in noncanonical
variables of importance to plasma physics, namely the integration
of magnetic field lines and the guiding center equations, can be obtained
via an action of this restricted noncanonical form. We have shown
how to write an extended phase space action for this class of noncanonical
variables with nonuniform time stepping. We have developed discretizations
that lead again to DVIs. The generalization of the extended phase space method, to noncanonical variables and to the second
order accurate DVI schemes introduced in this paper, is straightforward. This capacity
for nonuniform time stepping leads immediately to the capability for adaptive time stepping,
as described for symplectic integrators in Ref.~\onlinecite{Richardson_2011}.


\section*{Appendix A: Detailed proofs of the DVI single-step property}\label{AppendixA}
\begin{theorem}[linearized single-step property]\label{l_onestep_appendix}
Let $L_d(z_1,z_2)$ be a properly-degenerate discrete Lagrangian, and introduce the $m\times m$- matrices $[A(z_1,z_2)],[B(z_1,z_2)],[C(z_1,z_2,z_3)]$ with components
\begin{align}
A_{ij}(z_1,z_2) =& \mathcal{M}_{ji}(z_1,z_2)\\
B_{ij}(z_1,z_2) = &\mathcal{M}_{ij}(z_1,z_2)\\
C_{ij}(z_1,z_2,z_3) = & \frac{\partial L_d}{\partial z_2^i\partial z_2^j}(z_1,z_2) + \frac{\partial^2 L_d}{\partial z_1^i \partial z_1^j}(z_2,z_3).
\end{align} 
Under the following transversality assumptions,
\begin{enumerate}
\item[\emph{(G1)}] For each $(z_1,z_2),(z_1^\prime,z_2^\prime)\in Z\times Z$ near the diagonal, $\text{\emph{im}} [A(z_1,z_2)]\cap \text{\emph{im}} [B(z_1^\prime,z_2^\prime)] = \{0\}$.
\item[\emph{(G2)}] For each $(z_1,z_2),(z_1^\prime,z_2^\prime)\in Z\times Z$ near the diagonal in $Z\times Z$ and $(z_1^{\prime\prime},z_2^{\prime\prime},z_3^{\prime\prime})\in Z\times Z\times Z$ near the diagonal in $Z\times Z\times Z$, 
\[
[C(z_1^{\prime\prime},z_2^{\prime\prime},z_3^{\prime\prime})](\text{\emph{ker}}[B(z_1^\prime,z_2^\prime)])
\]
is a graph over $\text{\emph{im}}[A(z_1,z_2)]$,
\end{enumerate}
the discrete Euler-Lagrange equations linearized about a trajectory $k\mapsto z^0_k$ whose neighboring samples satisfy $|z_{k+1}^0-z_{k}^0| < \delta$ for some small $\delta >0$ independent of $k$  are equivalent to a single-step method.
\end{theorem}

\begin{remark}
If $L_d$ is some properly-degenerate discrete Lagrangian satisfying (G1) and (G2), then all properly-degenerate discrete Lagrangians in a neighborhood of $L_d$ will satisfy (G1) and (G2). In practice this observation greatly simplifies the task of verifying (G1) and (G2) because the $h\rightarrow 0$ limit of a properly-degenerate discrete Lagrangian is usually quite simple to analyze. 
\end{remark}

\begin{remark}
The condition $|z_{k+1}^0-z_{k}^0| < \delta$ is generally satisfied provided that the timestep $h$ in a variational integrator is sufficiently small.
\end{remark}

\begin{proof}
The proof picks up at the end of the proof sketch of Theorem \ref{l_onestep}.

To that end, consider the linear map $\Phi:\mathbb{R}^N\rightarrow X_{k+1}\times Y_k$ given by
\begin{align}
\Phi(\delta z) = ([\pi_X(k+1)][C(k+1)]\delta z,[B(k)]\delta z).
\end{align}
By Eqs.\,\eqref{lprf_one}-\eqref{lprf_two} it is enough to show that the kernel of $\Phi$ is trivial. To see that this is so, first note that by transversality assumption (G2) the linear space $[C(k+1)](\text{ker}[B(k)])$ must be of the form
\begin{align}
&[C(k+1)](\text{ker}[B(k)]) \nonumber\\
&= \{w_X+\Gamma(w_X)\mid w_X\in \text{im}[A(k+1)]\},
\end{align}
where $\Gamma:\text{im}[A(k+1)]\rightarrow \text{im}[B(k+1)]$ is a linear map. In particular, $\text{dim}\,[C(k+1)](\text{ker}[B(k)]) = \text{dim}\,\text{im}[A(k+1)] = m/2$, which by the rank-nullity theorem implies that $[C(k+1)]\mid \text{ker}[B(k)]$ is invertible onto its image.
Now suppose that $\Phi(\delta z) =0 $. This implies that $\delta z$ must be in the kernel of $[B(k)]$. Therefore $[C(k+1)]\delta z\in [C(k+1)](\text{ker}[B(k)])$ must have the form
\begin{align}
[C(k+1)]\delta z = w_X + \Gamma(w_X),
\end{align}
for a unizue $w_X\in \text{im}[A(k+1)]$. But because $[\pi_X(k+1)][C(k+1)]\delta z=0$, it must be the case that
\begin{align}
0 = & [\pi_X(k+1)][C(k+1)]\delta z\nonumber\\
=& [\pi_X(k+1)](w_X + \Gamma(w_X))\nonumber\\
=& w_X,
\end{align}
which implies that $\delta z = 0$.
\end{proof}

\begin{theorem}[nonlinear single-step property]\label{nl_onestep_appendix}
Let $L_d(z_1,z_2)$ be a properly-degenerate discrete Lagrangian that satisfies the transversality conditions \emph{(G1)} and \emph{(G2)} given in the statement of Theorem \ref{l_onestep}. Solutions $k\mapsto z_k$ of the discrete Euler-Lagrange equations near a given solution $k\mapsto z_k^0$ that satisfies $|z_{k+1}^0-z_{k}^0|< \delta$ for some sufficiently small $k$-independent $\delta >0$ are generated by a single-step method $\varphi:Z\rightarrow Z$. In other words,
\[
z_{k+1} = \varphi(z_k)
\]
for each $k$.
\end{theorem}

\begin{proof}
First we introduce some convenient notation. Let $\alpha:Z\times Z\rightarrow \mathbb{R}^m$ and $\beta:Z\times Z\rightarrow\mathbb{R}^m$ be the functions defined by
\begin{align}
\alpha_i(z_1,z_2) =& \frac{\partial L_d}{\partial z_2^i}(z_1,z_2)\\
\beta_i(z_1,z_2) = & \frac{\partial L_d}{\partial z_1^i}(z_1,z_2).
\end{align}
For each $\tilde{z}\in Z$, also define the related functions $\alpha_{\tilde{z}}:Z\rightarrow \mathbb{R}^m$ and $\beta_{\tilde{z}}:Z\rightarrow \mathbb{R}^m$ according to 
\begin{align}
\alpha_{\tilde{z}}(z) = & \alpha(z,\tilde{z})\\
\beta_{\tilde{z}}(z) = & \beta(\tilde{z},z).
\end{align}
Finally, introduce the discrete Euler-Lagrange operator $E:Z\times Z\times Z\rightarrow \mathbb{R}^m$ given by
\begin{align}
E(z_1,z_2,z_3) = \alpha(z_1,z_2) + \beta(z_2,z_3),
\end{align}
and the associated function $E_z:Z\times Z\rightarrow\mathbb{R}^m$ given by
\begin{align}
E_z(z_1,z_2) = E(z_1,z,z_2).
\end{align}
In terms of these notations, the discrete Euler-Lagrange equations may be written in several equivalent ways:
\begin{align}
0 =& E(z_{k-1},z_k,z_{k+1}) \\
=& E_{z_{k}}(z_{k-1},z_{k+1})\\
=& \alpha_{z_k}(z_{k-1}) + \beta_{z_k}(z_{k+1})\\
=& \alpha(z_{k-1},z_k) + \beta(z_k,z_{k+1}).
\end{align}

By the constant-rank theorem, for each $z$ the level sets of either $\alpha_z$ or $\beta_z$ are $m/2$-dimensional submanifolds that foliate $Z$. We will call a level set of $\alpha_z$ an $\alpha$-\emph{leaf}, and a level set of $\beta_z$ a $\beta$-\emph{leaf}. We may choose mutually-disjoint neighborhoods $U_{k}$ of each $z_k^0$ such that the intersection of either the $\alpha$-foliation or the $\beta$-foliation with $U_k$
 is diffeomorphic to $\mathbb{R}^{m/2}\times\mathbb{R}^{m/2}$. In particular we may define smooth maps
\begin{align}
\gamma^\alpha_z: \cup_k U_k\rightarrow \mathbb{R}^{m/2}\\
\gamma^\beta_z: \cup_k U_k \rightarrow \mathbb{R}^{m/2}
\end{align} 
such that the restriction of $\gamma_z^\alpha$ ($\gamma^\beta_z$) to $U_k$ is a quotient map for the $\alpha$-foliation ($\beta$-foliation) intersected with $U_k$. Moreover, we may assume without loss of generality that $\gamma^\alpha_{z}(z_k^0) = 0 = \gamma^\beta_z(z_k^0)$, independent of $z$. 

Because, for each $z$, $\alpha_z$ is constant along the $\alpha$-leaves and $\beta_z$ is constant along the $\beta$-leaves, the discrete Euler-Lagrange operator $E(z_1,z,z_3)$ only depends on the $\alpha$-leaf that contains $z_1$ and the $\beta$-leaf that contains $z_3$. Therefore for each $z_k\in Z$ there must be a function $\varepsilon_{z_k}:\mathbb{R}^{m/2}\times\mathbb{R}^{m/2}\rightarrow \mathbb{R}^m$ defined by the relation 
\begin{gather}
E_{z_k}(z_{k-1},z_{k+1}) = \varepsilon_{z_k}(X, Y)\label{descend}\\
X = \gamma^\alpha_{z_k}(z_{k-1})\\
Y = \gamma^\beta_{z_k}(z_{k+1})
\end{gather} 
for $z_k\in U_k$.

By hypothesis (G1) the derivative $D\varepsilon_{z_k}(X,Y):\mathbb{R}^{m/2}\times\mathbb{R}^{m/2}\rightarrow \mathbb{R}^m$ is invertible for each $(X,Y)$ and $z_k\in U_k$. Therefore by the inverse function theorem the function $\varepsilon_{z_k}$ restricts to a diffeomorphism on a neighborhood of $(X_{k-1}(z_{k}),Y_{k+1}(z_k)) = (\gamma^\alpha_{z_k}(z_{k-1}^0),\gamma^\beta_{z_k}(z_{k+1}^0))$. At the price of possibly shrinking the $U_k$, we may assume that this neighborhood is all of $\mathbb{R}^{m/2}\times\mathbb{R}^{m/2}$.

Let $(\widehat{X}_{z_k},\widehat{Y}_{z_k}) = (\varepsilon_{z_k})^{-1} $ be the inverse of the diffeomorphism $\varepsilon_{z_k}:\mathbb{R}^{m/2}\times\mathbb{R}^{m/2}\rightarrow\mathbb{R}^m$. The discrete Euler-Lagrange equations $E(z_{k-1},z_k,z_{k+1})= 0$ imply
\begin{align}
0 =  \varepsilon_{z_k}(\gamma^\alpha_{z_k}(z_{k-1}),\gamma^\beta_{z_{k}}(z_{k+1})),
\end{align}
which is equivalent to
\begin{align}
\gamma^\alpha_{z_k}(z_{k-1}) =& \widehat{X}_{z_k}(0)\label{tt}\\
\gamma^\beta_{z_k}(z_{k+1}) =&\widehat{Y}_{z_k}(0).
\end{align}
In particular shifting Eq.\,\eqref{tt} gives the $m$ equations for the $m$ unknowns $z_{k+1}$
\begin{align}
F^\alpha(z_{k},z_{k+1})&=\gamma^\alpha_{z_{k+1}}(z_{k})- \widehat{X}_{z_{k+1}}(0)=0 \\
F^\beta(z_k,z_{k+1})&=\gamma^\beta_{z_k}(z_{k+1})-\widehat{Y}_{z_k}(0) =0.
\end{align}
The proof will therefore be complete if we can show that the mapping $z_{k+1}\mapsto (F^\alpha(z_k,z_{k+1}),F^\beta(z_k,z_{k+1}))$ is a diffeomorphism for fixed $z_k$ in a neighborhood of $z_{k+1}^0$.

To that end, note that by the implicit function theorem it is enough to show that the linear map
\begin{align}
&\Phi:\delta z_{k+1}\nonumber\\
&\mapsto \bigg(D_{z_{k+1}}F^\alpha(z_k^0,z_{k+1}^0)[\delta z_{k+1}], D_{z_{k+1}}F^\beta(z_k^0,z_{k+1}^0)[\delta z_k] \bigg)
\end{align}
has trivial kernel. Demonstrating that this is so amounts to reproducing the proof of Theorem \ref{l_onestep}. The summary is the following.

Suppose that $\delta z_{k+1}$ is in the kernel. Because $\gamma^\alpha_{z_{k+1}}(z_{k}^0) = 0$ for each $z_{k+1}$, $\delta z_{k+1}$ must satisfy
\begin{align}
0=&\delta \widehat{X} = \frac{d}{d\epsilon}\bigg|_0 \widehat{X}_{z_{k+1}^0+\epsilon \delta z_{k+1}}(0)\label{ceqn}\\
0=&\delta \gamma^\beta = \frac{d}{d\epsilon}\bigg|_0 \gamma^\beta_{z_{k}^0}(z_{k+1}^0+\epsilon \delta z_{k+1})\label{tangent}.
\end{align}
The second equation \eqref{tangent} will be satisfied if and only if $\delta z_{k+1}$ tangent to the $\beta$-leaf passing through $z_k^0$. This means 
\begin{align}
B_{ij} \delta z_{k+1}^j = 0,
\end{align}
or $\delta z_{k+1}$ is in the kernel of the matrix $[B]$ defined in Eq.\,\eqref{B_shorthand}. Moving now to Eq.\,\eqref{ceqn},  note that because $\varepsilon_{z_{k+1}}(\widehat{X}_{z_{k+1}}(0),\widehat{Y}_{z_{k+1}}(0)) = 0$ by definition, differentiating in $z_{k+1}$ at $z_{k+1}^0$ gives
\begin{align}
0 = & D_X\varepsilon_{z_{k+1}^0}(0,0)[\delta \widehat{X}]+ D_Y\varepsilon_{z_{k+1}^0}(0,0)[\delta\widehat{Y}] \nonumber\\
&+ \frac{d}{d\epsilon }\bigg|_0 \varepsilon_{z_{k+1}^0+\epsilon \delta z_{k+1}}(0,0),\label{need_ds}
\end{align}
where we have used $(\widehat{X}_{z_{k+1}^0}(0),\widehat{Y}_{z_{k+1}^0})(0) = (0,0)$ by our normalization convention for $\gamma^\alpha,\gamma^\beta$, and we have introduced
\begin{align}
\delta \widehat{Y} = \frac{d}{d\epsilon}\bigg|_0 \widehat{Y}_{z_{k+1}^0+\epsilon \delta z_{k+1}}(0). 
\end{align}
Each of the derivatives in Eq.\,\eqref{need_ds} may be expressed in terms of derivatives of $E$ by implicitly differentiating Eq.\,\eqref{descend}, which leads to 
\begin{align}
[A] \delta X + [C]\delta z_{k+1} + [B]\delta Y = 0, \label{basic_eqn}
\end{align}
where $\delta X$ is any vector that satisfies 
\begin{align}
D\gamma^\alpha_{z_{k+1}^0}(z_{k}^0)[\delta{X}] = \delta \widehat{X},\label{proj_A}
\end{align}
$\delta Y$ is any vector that satisfies 
\begin{align}
D\gamma^\beta_{z_{k+1}^0}(z_{k+2}^0)[\delta{Y}] = \delta \widehat{Y},
\end{align} 
and the matrices $[A],[B],[C]$ are defined in Eqs.\,\eqref{A_shorthand}-\eqref{C_shorthand}. Now using $\delta \widehat{X} = 0$, Eq.\,\eqref{proj_A} implies that $\delta X$ must be in the kernel of $[A]$. Therefore if we apply the projection matrix $[\pi_X]$ guaranteed by transversality assumption (G1) to Eq.\,\eqref{basic_eqn}, we obtain
\begin{align}
[\pi_X][C]\delta z_{k+1} = 0.
\end{align}
But because $\delta z_{k+1}$ is in the kernel of $[B]$, transversality assumption (G2) implies that $\delta z_{k+1}=0$.

\end{proof}


\section*{Appendix B: Reversibility: a warning}\label{AppendixB}


We first consider
the one degree of freedom Hamiltonian
\begin{equation}
H=\frac{p^{2}+q^{2}}{2}+\frac{\alpha qp^{2}}{2}.\label{eq:ReversibleHamiltonian}
\end{equation}
Applying either the scheme in Eq.~(\ref{eq:Mod-SE-scheme}) 
or its adjoint in Eq.~(\ref{eq:Adj-Mod-SE-scheme}),
we find, of course,  first order accuracy but also good long time properties, the latter because of the preservation of the two-forms in Eqs.~\eqref{eq:Two-form-with-O(h)},\eqref{eq:Two-form-adjoint}. If we
compose the two schemes, we also find good long time properties, and
with second order accuracy. However, these favorable properties are
traced not to the preservation of a two-form but to the reversibility of the Hamiltonian in Eq.~(\ref{eq:ReversibleHamiltonian}):
the symmetry $R:(q,p)\mapsto(q,-p)$ leaves $H$ invariant, and the fixed points of this symmetry are $p=0$. This reversibility is inherited by the exact time-$h$ map $\Phi_h$, i.e.~$\Phi_h$  satisfies
\begin{align}
R\circ\Phi_{h}=\Phi_{h}^{-1}\circ R.\label{Map-reversibility}
\end{align}
If a discrete scheme $T_h$ also satisfies this map reversibility, it should have the favorable properties due to reversibility\cite{Finn_2015}. In fact, neither $T_h$ nor $T_h ^{\dagger}$ (where adjoint is defined as  $T_h^{\dagger}=T_{-h}^{-1}$; see Ref.~\onlinecite{Hairer_2006}) satisfy this map reversibility property. However, $T_h$ does satisfy the related property \emph{weak reversibility}\cite{Finn_2015}, 
\[
R\circ T_{h}=T_{-h}\circ R,
\]
and similarly for $T_h^{\dagger}$. From this property it follows that the composed scheme $\Psi_h=T_{h}\circ T_{h}^{\dagger}$
(or $T_{h}^{\dagger}\circ T_{h}$) is also weakly reversible, and because it is self-adjoint, it is also reversible\cite{Finn_2015}. 
This map reversibility appears to be responsible for the observed good long-time behavior.

As discussed in Refs.\,\onlinecite{Richardson_2011} and \onlinecite{Finn_2015}, it can be misleading
to evaluate a scheme by testing it on a reversible Hamiltonian system, because good
results might be obtained solely due to the reversibility property and not from any property inherited from the variational nature.

The Hamiltonian $H=(p^{2}+q^{2})/2+\alpha q p^{3}/3$
considered in Sec.~\ref{Numerical-no-go-composing} has another symmetry $(q,p)\mapsto(-q,-p)$, but this symmetry preserves the point $(q,p)=(0,0)$ rather than a line ($p=0$), and such a symmetry does not endow any special properties, so we do not consider this Hamiltonian to be reversible. See Refs.\,\onlinecite{Richardson_2011,Finn_2015}. And indeed, the results in  Sec.~\ref{Numerical-no-go-composing} show that the orbits spiral out, showing the lack of a preserved two-form.

\section*{Appendix C: Analog of the symplectic Euler scheme for the magnetic field line
problem\label{AppendixC}}

Here, we consider the most direct analog to the symplectic Euler scheme
for canonical variables, applied to the class of noncanonical systems of Sec.~\ref{sec:DVI2}. We specialize to the magnetic field line problem for concereteness, and have
\begin{equation}
L_{d}(x_{k}^{1},x_{k}^{2},x_{k+1}^{2})=A_{2}(x_{k}^{1},x_{k+1}^{2})(x_{k+1}^{2}-x_{k}^{2})+hA_{3}(x_{k}^{1},x_{k+1}^{2}).\label{L_d_MFL_SE}
\end{equation}
If this system has a preserved two-form with $\omega(-h)=\omega(h)$,
it can be composed with its adjoint to preserve $\omega$ and obtain second order accuracy.
Note that $L_d$ depends on $x^{1}$ at only one time level and therefore, as noted in Sec.~\ref{sec:DVI}, its
discrete Hessian has rank 1. This shows that the system is indeed properly degenerate, and is a DVI.

Its preserved two-form is found simply by looking at the endpoint
terms in $dS$ for $k=0,1$:
\[
S=A_{2}(x_{0}^{1},x_{1}^{2})(x_{1}^{2}-x_{0}^{2})+hA_{3}(x_{0}^{1},x_{1}^{2}),
\]
from which we find
\[
dS=\left[A_{2,1}(x_{0}^{1},x_{1}^{2})(x_{1}^{2}-x_{0}^{2})+hA_{3,1}(x_{0}^{1},x_{1}^{2})\right]dx_{0}^{1}
\]
\[
+\left[-A_{2}(x_{0}^{1},x_{1}^{2})\right]dx_{0}^{2}+\left[A_{2,1}(x_{1}^{1},x_{2}^{2})(x_{2}^{2}-x_{1}^{2})+hA_{3,1}(x_{1}^{1},x_{2}^{2})\right]dx_{1}^{1}
\]
\[
+\left[A_{2,2}(x_{0}^{1},x_{1}^{2})(x_{1}^{2}-x_{0}^{2})+hA_{3,2}(x_{0}^{1},x_{1}^{2})+A_{2}(x_{0}^{1},x_{1}^{2})-A_{2}(x_{1}^{1},x_{2}^{2})\right]dx_{1}^{2}
\]
\[
+\left[A_{2,2}(x_{1}^{1},x_{2}^{2})(x_{2}^{2}-x_{1}^{2})+hA_{3,2}(x_{1}^{1},x_{2}^{2})+A_{2}(x_{1}^{1},x_{2}^{2})-A_{2}(x_{2}^{1},x_{3}^{2})\right]dx_{2}^{2}
\]
\[
+\left[A_{2}(x_{2}^{1},x_{3}^{2})\right]dx_{2}^{2},
\]
with the last term subtracted in the $dx_{2}^{2}$ term. Because of
satisfying the discrete Euler-Lagrange equations, all the terms except
for the endpoint terms
\[
dS=A_{2}(x_{2}^{1},x_{3}^{2})dx_{2}^{2}-A_{2}(x_{0}^{1},x_{1}^{2})dx_{0}^{2}
\]
vanish, and $d^{2}S=0$ leads to the preservation of the two-form is
composed
\begin{equation}
\omega=dA_{2}(x_{0}^{1},x_{1}^{2})\wedge dx_{0}^{2}.\label{eq:MFL_SE_form}
\end{equation}
Upon substituting 
\[
x_{1}^{2}=x_{0}^{2}+hu(x_{0}^{1},x_{1}^{2}),\,\,\,\,u(x_{0}^{1},x_{1}^{2})=-\frac{hA_{3,1}(x_{0}^{1},x_{1}^{2})}{A_{2,1}(x_{0}^{1},x_{1}^{2})},
\]
equal to $hB_{2}(x_{0}^{1},x_{1}^{2})/B_{3}(x_{0}^{1},x_{1}^{2})$
and from the $dx_{0}^{1}$ term in $dS$, we find the preserved two-form
$\omega$ equals
\[
\left(A_{2,1}(x_{0}^{1},x_{1}^{2})+hA_{2,2}(x_{0}^{1},x_{1}^{2})u(x_{0}^{1},x_{1}^{2})+O(h^{2})\right)dx_{0}^{1}\wedge dx_{0}^{2},
\]
the flux invariant $B_{3}dx_{0}^{1}\wedge dx_{0}^{2}$
of the continuous system plus a $O(h)$ correction, proportional to $A_{2,2}$. As for the modified
symplectic Euler scheme of Sec.~\ref{Sec.Composing}, the adjoint of this scheme
preserves the same form but with $\omega(h)\to\omega(-h)\neq\omega(h)$,
and therefore the composition of this scheme with its adjoint cannot
be assured of having a preserved two-form. The essential difference between this scheme and the symplectic Euler scheme for a canonical system is the dependence of $A_2$ on $x^2$ in Eq.~\eqref{L_d_MFL_SE}.




\section*{Acknowledgments}
This material is based upon work supported by the National Science Foundation under Grant No.\,1440140, while the authors JB and JMF were in residence at the Mathematical Sciences Research Institute in Berkeley, California, during the fall semester of 2018. Research presented in this article was supported by the Los Alamos National Laboratory LDRD program under project number 20180756PRD4. A portion of this work was performed under the auspices of the U.S. Department of Energy by Lawrence Livermore National Laboratory under Contract DE-AC52-07NA27344.
\bibliography{cumulative_bib_file_older.bib}



\end{document}